\documentclass[11pt]{article}
\usepackage[a4paper,
left=1in,
right=1in,
top=1in,
bottom=1in,
]{geometry}

\usepackage{booktabs} 

\usepackage{natbib}

\usepackage{authblk} 

\usepackage[utf8]{inputenc} 
\usepackage[T1]{fontenc}    
\usepackage{amsthm}         
\usepackage{amsmath}        
\usepackage{amsfonts}       
\usepackage{hyperref}       
\usepackage{cleveref}       
\AddToHook{cmd/appendix/before}{
    \crefalias{section}{appendix}
    \crefalias{subsection}{appendix}
}
\usepackage{subcaption}		
\usepackage{pgfplots}		
\usepackage{thm-restate}	
\usepackage{mathtools}		
\usepackage{xspace}         

\usepackage{tikz}
\usetikzlibrary{calc,decorations.pathreplacing}
\tikzset{ 
	protovertex/.style={
		draw,
		circle,
		inner sep=0,
		minimum size=.15cm}
}
\usepgfplotslibrary{statistics}

\definecolor{TUMBlue}{HTML}{0065BD}

\usepackage{todonotes} 
\presetkeys{todonotes}{inline}{}

\newcommand{\fwnote}[1]{\todo[color=green!60]{\textbf{Future Work:}\\#1}}
\renewcommand{\fwnote}[1]{}

\newtheorem{theorem}{Theorem}[section]

\newtheorem{lemma}[theorem]{Lemma}

\newtheorem{corollary}[theorem]{Corollary}

\renewcommand{\log}{\ln}

\newcommand{\fl}{\mathbf{L}}
\newcommand{\fw}{\mathbf{W}}

\newcommand{\E}{\mathbb{E}}
\newcommand{\p}{\mathbb{P}}
\newcommand{\N}{\mathbb{N}}
\newcommand{\Z}{\mathbb{Z}}
\newcommand{\R}{\mathbb{R}}
\newcommand{\eps}{\varepsilon}

\newcommand{\loss}[1][ALG]{\fl_{\text{#1}}}
\newcommand{\waitobj}[1][ALG]{\fw_{\text{#1}}}

\newcommand{\alg}{\text{ALG}}
\newcommand{\gdy}{\text{GDY}}

\newcommand{\pat}{\text{PAT}}
\newcommand{\dep}{\mu}
\newcommand{\Tmax}{T} 
\newcommand{\distname}{sojourn distribution\xspace}
\newcommand{\distnamepl}{sojourn distributions\xspace}
\newcommand{\dtimename}{sojourn time\xspace}
\newcommand{\dtimenamepl}{sojourn times\xspace}
\newcommand{\dparaname}{degree parameter\xspace} 
\newcommand{\dparanamepl}{degree parameters\xspace} 
\newcommand{\sou}{X}
\newcommand{\arr}{A}
\newcommand{\numbarr}{a}
\newcommand{\pool}{Z}
\newcommand{\psize}{z}
\newcommand{\ag}{v}
\newcommand{\agtwo}{\hat v}

\newcommand{\tiebreak}{R} 

\usepackage{todonotes} 
\presetkeys{todonotes}{inline}{}

\renewcommand*{\thefootnote}{\fnsymbol{footnote}} 

\title{Superiority of Instantaneous Decisions in Thin Dynamic Matching Markets}

\author[1]{Johannes Bäumler}
\author[2]{Martin Bullinger}
\author[3]{Stefan Kober}
\author[4,5]{Donghao Zhu}

\affil[1]{ \small Department of Mathematics, University of California Los Angeles, USA}
\affil[2]{ \small School of Engineering Mathematics and Technology, University of Bristol, UK}
\affil[3]{ \small Department of Mathematics, Université Libre de Bruxelles, Belgium}
\affil[4]{ \small The Institute of Statistical Mathematics, Tokyo, Japan}
\affil[5]{ \small The University of Tokyo Market Design Center, Tokyo, Japan\protect\\ \vspace*{0.05cm} jbaeumler@math.ucla.edu, martin.bullinger@bristol.ac.uk, stefan.kober@ulb.be, zhu.donghao@ism.ac.jp}

\date{}

\begin{document}
	
	\maketitle

\begin{abstract}
	We study a dynamic matching setting where homogeneous agents arrive at random according to a Poisson process and randomly form edges yielding a sparse market. 
	Agents stay in the market according to a certain \distname and wait to be matched with a compatible agent by a matching algorithm. 
    When their maximum \dtimename is reached, they perish unmatched.
	The primary objective is to maximize the number of matched agents. 
	Our main result is to show that a uniformly guaranteed \dtimename suffices to get almost optimal performance of instantaneous matching.
    Interestingly, this matching policy essentially keeps as few agents in the market as possible.
    Hence, in contrast to the common paradigm that market thickness is the crucial property for obtaining strong matching performance, we show that the agents' sojourn behavior can be an equally powerful factor. 
	In addition, instantaneous matching is close to optimal with respect to minimizing waiting time. 
	We develop new techniques for proving our results going beyond commonly adopted methods for Markov processes.\\
{\bf Keywords:} \texttt{Market design}, \texttt{matching markets}, \texttt{dynamic matching}
\end{abstract}

\renewcommand*{\thefootnote}{\arabic{footnote}}

\section{Introduction}

Matching is one of the most vibrant research areas at the intersection of economics, operations research, and computer science with an abundance of applications encompassing labor markets, school choice, child adoption, dating platforms, ad allocation, ride-sharing, or kidney exchange. 
In reality, many matching scenarios have a dynamic flavor in the sense that agents arrive and get matched over time.
Moreover, there might be uncertainty over when agents arrive in the market and for how long they stay.
We contribute to the research on algorithmic decision making in dynamic stochastic matching markets.

We follow the models by \citet{AAGK17a} and \citet{ALG20a} where agents arrive over the course of a continuous time horizon according to a Poisson process.
Agents then stay in the market for a random (maximum) \dtimename $\sou$.
They are compatible with other agents according to independent Bernoulli variables similar to the random graph model by \citet{ErRe60a}.
During her sojourn, an algorithm can match an agent with a compatible agent and the pair leaves the market.
However, if an agent is unmatched at the end of her sojourn, she perishes, i.e., leaves the market unmatched.
The time line of our matching market is illustrated in \Cref{fig:procedure}.
We want our matching policies to achieve two goals:
\begin{enumerate}
    \item Small loss, i.e., few unmatched agents.
    \item Small waiting time of agents until being matched.
\end{enumerate}

\begin{figure}
    \centering
	\begin{tikzpicture}
		\pgfmathsetmacro\rightdist{13.5}
		\pgfmathsetmacro\arrive{1}
		\pgfmathsetmacro\match{6}
		\pgfmathsetmacro\depart{10.5}

		\draw[->, thick] (0,0) -- (\rightdist,0);
		\node at ($(\rightdist,-.3)$) {\footnotesize continuous time};
		\node[anchor=north] at ($(\arrive,-.2)$) {\small random arrival};
		\node[text width=2.8cm,anchor=north,align = center] at ($(\depart,-.2)$) {\small perishing without being matched};
		\node[text width=3.5cm,anchor=north,align = center] at ($(\match,-.2)$) {\small matching opportunity during sojourn};
		
		\foreach \i in {\arrive,\depart}{
		\draw ($(\i,0) + (0,-.15)$) -- ($(\i,0) + (0,.15)$);}
		\draw[gray] ($(\match,0) + (0,-.15)$) -- ($(\match,0) + (0,.15)$);
		\draw [decorate,decoration={brace,amplitude=5pt}]
		(\arrive,.3) -- (\depart,.3) node[midway, yshift = .2cm,anchor = south]
        {\small random sojourn};
		
		\draw[very thick, ->] (\arrive,1.2) -- (\arrive,.5);
		\node at (\arrive,1.4) {\small \color{TUMBlue}{greedy}};
		\draw[very thick, ->] (\depart,1.2) -- (\depart,.5);
		\node at (\depart,1.4) {\small \color{TUMBlue}{patient}};
	\end{tikzpicture}
	
    \caption{Timeline of an agent's sojourn in our stochastic matching market.
    The greedy and patient matching policies are highlighted in blue.
    \label{fig:procedure}}
\end{figure}

Following \citet{ALG20a}, we study two matching policies: the \emph{greedy algorithm} $\gdy$ matches agents as early as possible, i.e., it checks for compatible agents at an agent's arrival and matches immediately.
Hence, every pair of agents waiting in the market is incompatible.
On the other extreme, the \emph{patient algorithm} $\pat$ matches agents only when perishing is imminent.
It therefore needs information of the perishing time of an agent, at which it gets a last opportunity to match the agent.
This approach seems sensible when aiming at small loss because an agent only perishes if she is incompatible to all agents in the pool at the end of her sojourn.
In fact, the main result by \citet{ALG20a} says that $\pat$ has a significantly smaller loss than $\gdy$ (exponentially small vs inverse linear) when the \dtimename 
is exponentially distributed.\footnote{Their bounds on the loss are parameterized by a \dparaname defined as the arrival rate of compatible agents for a fixed agent in the market.
We follow this convention when describing our results.
Note that \citet{ALG20a} use the terminology ``density parameter.''
We prefer \dparaname because it measures the expected number of edges in the compatibility graph arriving in a time unit, while the edge \emph{density} gets small for large markets under a fixed \dparaname.
}
However, our paper shows that this result is sensitive to the distribution of the \dtimename, 
and that the greedy algorithm need not be worse.
If agents are guaranteed to stay in the market for a small time, then the loss of $\gdy$ is exponentially small as well.

\begin{theorem}[Informal version of \Cref{thm:epslowerbound}]\label{thm:lossinformal}
    If the \dtimename of agents exceeds any uniform lower bound, then the loss of $\gdy$ is exponentially small.\footnote{There, the exponent constant depends on the minimum sojourn guarantee.}
\end{theorem}

We believe that the condition of a minimum \dtimename is reasonable in various matching scenarios. 
Indeed, it feels unlikely for agents to leave a market immediately, i.e., before there was a chance to match them at all.
Additionally, in many centralized matching markets, binding participation even mechanically prevents immediate exit. 
Classic examples include the National Resident Matching Program for assigning doctors to hospitals in the US, where participants are bound by participation agreements and cannot accept or solicit outside offers once inside the mechanism \citep{RoPe99a}.
While in such settings, matching is often performed in coordinated batches, our results suggest that immediate matching can be performed continuously if contracts merely enforce a modest participation duration. 
Moreover, a minimum sojourn can be further encouraged via standard commitment devices, such as a non-refundable deposit or a cancellation fee.
A large empirical literature shows that even small stakes meaningfully raise commitment: for example, deposit contracts improve smoking cessation and weight-loss adherence \citep{GKZ10a,VJT+08a} and commitment savings increase bank account balances \citep{AKY06a}.\footnote{On the theoretical side, the effectiveness of imposing a toll on agents has also been observed in queuing theory \citep{Naor69a}.}

Our next result states that the loss guarantee from \Cref{thm:lossinformal} is tight in the following sense.

\begin{theorem}[Informal version of \Cref{thm:lossgenlb}]
    For any sojourn distribution with finite mean, the loss of any algorithm is at least exponential.
\end{theorem}

By contrast, we also carve out the reason for the higher, inverse linear loss derived by \citet{ALG20a} when the \dtimename 
is exponentially distributed.
If there is too much density of the sojourn distribution around~$0$, a higher loss is inevitable.

\begin{theorem}[Informal version of \Cref{thm:highloss}]
    If the \distname can be lower-bounded by any linear function in any neighborhood of~$0$, then the loss of $\gdy$ is at least inverse linear.
\end{theorem}

Note that this theorem includes the exponential distribution.

For our second objective concerning the waiting time of agents, we obtain an inverse linear waiting time of $\gdy$, regardless of the sojourn distribution.

\begin{theorem}[Informal version of \Cref{lem:agg_waiting}]
    For any sojourn distribution, the average waiting time caused by $\gdy$ is inverse linear.
\end{theorem}

Up to a constant, this is once again the best any algorithm can achieve (cf.~\Cref{prop:waitlb}).\footnote{Note that \Cref{prop:waitlb} needs two cosmetic assumptions that we discuss in \Cref{sec:waiting time}.}
To summarize, the greedy algorithm is essentially optimal for both objectives.

Our results are surprising for several reasons.
A well-known cause for achieving small loss in matching markets is to establish and maintain a \emph{thick} market, i.e., a market where a large proportion of agents is present.
This is the intuition why $\pat$ should yield small loss.
More generally, the idea of delaying matching decisions to thicken the market is a paradigm that is repeatedly promoted in the matching literature \citep[see, e.g.,][and our discussion in the related work]{ShSm01a,EKW16a,ALG20a,BLY20a,Lesh22a,LMT22a}.
Even papers that demonstrate a strong performance of instantaneous matching algorithms base their success on market thickness  \citep{Unve10a, ABJM19a, ANS22a}. 

However, our result is qualitatively different.
In our model, the greedy algorithm sustains a \emph{thin} market, i.e., the fraction of agents kept in the market is small.
More precisely, it is an inverse linear function, which is as thin as markets can be in our stochastic model (cf. \Cref{prop:thickness}).
While our results consider a large market limit with a growing arrival rate of agents, we avoid generating a thickness in matching opportunities by fixing the \dparaname, which measures the arrival rate of compatible agents. 
Consequently, the compatibility probability, and thus the edge density in the compatibility graph, tends to zero.
The resulting market therefore additionally remains \emph{sparse} even as it becomes larger.\footnote{Note that a fixed \dparaname and growing arrival rate is akin to the setup by \citet{ALG20a}.}
Our work demonstrates that agents' sojourn behavior can drive the performance of a matching policy as much as market thickness does.

Moreover, simultaneously achieving a small loss and a small waiting time may seem to be conflicting goals.
For example, \citet{MNP20a} even obtain a necessary trade-off between matching quality and waiting time.
However, they operate on a thick market, and their result does not apply in our model.
It is thus the good performance of the greedy algorithm in a thin market that facilitates a best-of-both-worlds result.

Next, we want to zoom in on the patient algorithm.
A natural question is to compare the performance of $\gdy$ to that of $\pat$ beyond the exponentially distributed \dtimename studied by \citet{ALG20a}.
First, we would expect that $\pat$ maintains a strong performance in terms of loss minimization.
We confirm this by proving an exponentially small loss for constant \dtimenamepl (cf. \Cref{lem:patient_lb}).
The proof of this result is based on a comparison with an urn model, which is fundamentally different to the proof for an exponential \dtimename by \citet{ALG20a}.
However, in terms of waiting time, $\pat$ is clearly inferior.
Intuitively, it has to cause long waiting times as matching a pair has to be initiated by some agent exhausting her maximum sojourn. 
So half of the matched agents (as well as all of the agents leaving the market unmatched) have to exhaust their maximum sojourn.
We leverage this idea and show that $\pat$ typically\footnote{We only need an assumption on the sojourn distribution that guarantees that the probability of instantaneous criticality, i.e., $\p(\sou = 0)$, is not too large.} needs a constant average waiting time, only dependent on the distribution of the \dtimename close to~$0$ (cf.~\Cref{prop:pattime}).

We complement our theoretical results by a series of simulations. 
As discussed before, our theoretical findings consider the limit case for growing arrival rates.
However, the error terms vanish at rapid rates. 
Additionally, our simulations confirm our results for a rather small arrival rate indicating their robustness. 
Moreover, we test further specifications of the greedy algorithm. In both of our algorithms, the decision to match a pair is done by selecting a compatible partner uniformly at random.
This simple heuristics 
is enough to prove our guarantees.
We perform simulations for the case of a tie-breaking in favor of partners that have the smallest remaining sojourn, finding that this still yields a small improvement.
Additionally, our simulations demonstrate a surprising phenomenon: under a constant \dtimename, the expected loss of the greedy mechanism and patient algorithms seem to be identical. 
In \Cref{sec:conclusion}, we shed light on theoretical reasons for this observation.
In this section, we also discuss extensions of our work to the case of heterogeneous agents.

On the technical side, we develop novel techniques to obtain our results. 
While most of the existing literature relies on a steady state analysis of the pool size \citep[in particular][]{ALG20a,AAGK17a,ANS22a}, this quantity is not Markovian for general sojourn distributions, and we do not make use of a stationary distribution inherent to our process.
Instead, we perform a detour by a careful direct analysis obtaining uniform bounds over time for the pool size which in turn lead to bounds for the loss.

\section{Related Literature}

In this section, we review related literature.

\subsection{Modeling Matching Markets}

The rich literature on matching markets comprises a large number of formal models. 
These differ in their assumptions on the nature of agents, 
compatibility prerequisites, 
procedural specifications, and regarding the measurement of the output quality.
First, agents may be homogeneous or heterogeneous. 
The latter is often due to having certain applications in mind. 
For instance, agents can be partitioned into two classes like workers and firms in labor markets, children and schools in school choice, or children and potential adoptive parents in adoption markets. 
Moreover, agents might assume various additional properties such as being hard or easy to match, which play for instance an important role in kidney exchange [\citealp{ABJM19a}, \citealp{ANS22a}, see also \citealp{LMT22a} for a more general trade market model with superior and inferior buyers and sellers]. 
Second, there might be various constraints on whether agents are compatible. 
For instance, it might be prohibited that agents of the same class match (which is very reasonable for the above examples of \emph{bipartite} instances), or the possible matches could be given endogenously by a compatibility relation (which might even be deduced from an exogenous factor like affiliation to some type).
Third, procedural specifications concern the chronological process of the model. Agents might arrive stochastically, in fixed time steps, or even according to an adversary. 
Then, they stay in the market for some time after which they depart according to some procedure. 
Some stylistic models even assume an indefinite sojourn of agents until they are matched \citep{AAGK17a}.
Finally, we have to measure the quality of matchings. The two most common approaches considered in the literature are measuring the quality of a matching simply by its cardinality or by maximizing a more complex predefined objective function, which is often defined specifically for the given model and the application in mind. 
The former is appropriate in models where sparse markets, like the one in our paper, evolve from a compatibility relationship of the agents.
Another reasonable setting for this assumption is kidney exchange where losing an agent is severe.
By contrast, other objective functions require a careful consideration of model parameters and may quickly lead to complex measures \citep[see, e.g.,][]{EKW16a,AACCGKMWW17,PSST22a}.
Objective functions usually include the quality of matches, but can also consider waiting time, i.e., loss suffered by spending time in the market.
Hence, they are a combination of the goals that we study as well.

\subsection{Dynamic Matching Markets}

We most closely follow the stochastic and dynamic market models proposed by \citet{AAGK17a} and \citet{ALG20a}. 
We extend their models by considering a variable sojourn distribution and get new insights in the circumstances under which greedy matching performs well.
\citet{AAGK17a} seek to model barter exchange (having in mind the important application of kidney exchange) and the model is essentially a directed version of the model by \citet{ALG20a}. 
New agents form unilateral relationships (or directed edges) with probability~$p$. 
However, rather than stochastically, they arrive in fixed time intervals and for an indefinite sojourn, and the matching technology even allows for matching a larger group of agents or chains originating from altruistic donors. 
Their goal is to minimize total waiting time and the main result is that allowing more matching possibilities significantly reduces waiting time. If the matching technology is restricted to pairs, then their model closely resembles our model, where compatibility is interpreted as having both directed edges present, amounting to a compatibility probability of $p^2$. 
The waiting time in their result matches the waiting time that we observe for the greedy algorithm.
Therefore, we complement the result by \citet{AAGK17a} for different procedural specifications. 
\citet{NaZh21a} consider a bipartite version of the model by \citet{ALG20a}, where the agents of the two partition classes arrive with two different Poisson rates.
They also consider the greedy and patient algorithm. 
Performing a Markov chain analysis, they confirm the superiority of the patient algorithm in a two-sided market, but they observe a similar performance of the two algorithms for a one-sided market where one class of the agents are merely objects.

Much of the literature on dynamic versions of matching markets originates from research in theoretical computer science on the online formation of bipartite matchings \citep{KVV90a}. 
There, the primary goal is to maximize the cardinality of the matching formed instantaneously, while procedural specifications are secondary and algorithms are measured by their performance against an adversarial arrival of agents. 
With the purpose of modeling ad allocation, the model was generalized by \citet{MSVV07a}. 
Stochastic versions of the model where a fixed set of agents arrives in a uniformly random order were studied in depth \citep{FMMM09a,MGS12a,KRTV13a,EFGT22a,BRS26a}.
An adversarial arrival of agents is also considered by \citet{EKW16a} who consider an objective function combining the quality of matched agents and their waiting time, and by \citet{ABD+19a} who consider a model with arrivals and departures after constant time.

\subsection{Market Thickness and Causes of Strong Matching Performance}

The paradigm of waiting to create further possibilities, i.e., the desire for market thickness, observed by \citet{ALG20a} is ubiquitous in the literature. 
This phenomenon often occurs in markets with heterogeneous types and unbounded sojourn of agents, and we identify two main circumstances for its appearance. 
The first encompasses markets with the concurrent objectives to maximize matching quality while minimizing waiting time, and the second concerns settings, where agents are offered matching opportunities which they can accept or deny, similar to the secretary problem \citep{Ferg89a,Bear06a}.

Regarding the former, \citet{BLY20a} model child adoption by a bipartite matching market with two agent types and identify thresholds to decide when to conduct low-quality matches. In a similar vein, \citet{LMT22a} consider a bipartite trade model with a superior and inferior type of buyers and sellers, and determine how many agents to keep in the market. Interestingly, they find that few stored agents lead to a good market performance, but this result has to be interpreted with respect to their market design which requires to match many agents early. \citet{BRS+22a} consider a bipartite market with exponentially distributed \dtimenamepl. 
Their focus is on the distribution of matching values and they identify certain threshold rules which achieve a desired market thickness.
In models with adversarial arrival of agents, \citet{EKW16a}, \citet{CIL+20a}, and \citet{PSST22a} consider the fine-tuned timing of matching decisions.

In the second type of markets, the challenge is to balance information gains through rejecting matches with the danger of turning down a promising match. 
There, \citet{Lesh22a} studies a queuing model applied to an assignment problem, and finds that it is beneficial to decline mismatches. This can have both positive aspects for the declining agent who might obtain a better object, and for other agents in the market who might get served earlier. 
Furthermore, \citet{ACT21a} consider a matching market with arrivals and departures in discrete time steps and investigate cutoff strategies for accepting a match. 
Loosely related, \citet{ShSm01a} consider a model for labor markets where heterogeneous agents incur a cost for searching a suitable match. 
They find that agents of high productivity should wait to improve their matching quality. 
Also, note that the matching problem in the above paper by \citet{PSST22a} is equivalent to an adversarial version of the secretary problem.

There is far less evidence for the antagonistic principle of taking instantaneous decisions, with the notable exception of the research on kidney exchange, originating from seminal contributions by \citet{RSU04a,RSU05a}. 
In this line of research, agents usually stay indefinitely in the market and the performance is measured by means of waiting time. \citet{Unve10a} anticipates the conclusion by \citet{AAGK17a}, namely the optimality of greedy matching. 
In his model, the role model for \citet{AAGK17a}, he considers arrivals by a Poisson process and an indefinite sojourn, but the model contains agents of homogeneous types and the waiting time objective is specifically designed for the given setting. Follow-up work by \citet{ABJM19a} offers a closer look at the case of easy-to-match and hard-to-match agents and their prioritization in greedy-type algorithms. These theoretical findings are in accordance with simulations based on real-world data \citep{ABB+18a}. \citet{ANS22a} show the optimality of greedy matching in a similar setting. All of this work has in common that the reason for the optimality of greedy matching is market \emph{thickness}.
By contrast, our work presents the optimality of an instantaneous algorithm measured by the \emph{loss} of agents in a sparse market that is \emph{thin} (for most of the time). 
Recently, \citet{KAG25a} consider a model with heterogeneous types whose compatibility is globally fixed. They show that, under certain conditions on the compatibility structure, greedy matching is optimal with respect to regret.

Finally, \citet{ADSW25a} consider a bipartite model with heterogeneous agents where the objective function is a tradeoff of matching value and waiting costs. Since the model is quite complex model, a Markovian analysis is not feasible. 
They circumvent this obstacle by finding optimal matching rates by solving an optimization problem, and identify an algorithm that mimics these rates. 
Interestingly, this paper is the only related paper we are aware of which investigates the influence of the \dtimename.

\section{Model}\label{sec:model}

In this section, we present our formal model.
First, we introduce our model for a dynamic matching market, and then discuss our objectives.

\subsection{Matching Markets}
Following the models by \citet{AAGK17a} and \citet{ALG20a}, we consider a continuous-time matching market in the time window $[0,\Tmax]$ for some \emph{maximum time} $\Tmax>0$. 
Agents arrive with Poisson rate $m$ and enter the market. 
We label the agents $v_1, v_2,\ldots$ in the order of their arrival. 
Given times $s,t\in [0,\Tmax]$ with $s<t$, we denote by $\arr_t$ the set of \emph{agents arriving at time $t$} and by $\arr_{\left[s,t\right]}$ and $\numbarr_{\left[s,t\right]}$ the set and number of agents arriving in the time interval $\left[s,t\right]$, respectively. 
We denote by $\pool_t$ the (random) set of agents present \emph{after} processing all events\footnote{These include random events like the arrival of an agent as well as algorithmic decisions.} up to and including time $t$ and refer to it as the \emph{pool} at time $t$.
We write $\pool_{t-}$ for the pool \emph{just before} processing events at $t$.

Hence, upon her arrival, an agent enters the pool unless she gets matched immediately.
We denote the size of $\pool_t$ by $\psize_t$, i.e., $\psize_t = |\pool_t|$. 
We assume that the market is initially empty, i.e., $\pool_0 = \emptyset$. 

While we will later consider a growing arrival rate, we want to maintain a market with similar matching opportunities.
This is governed by a \emph{\dparaname}~$d$.
Once an agent enters the pool, we decide on her \emph{compatibility} with agents already in the pool by independent Bernoulli random variables with parameter $p = \frac dm$. 
Hence, for each agent in the pool, compatible agents arrive at Poisson rate $d = p \cdot m$. 
Formally, we model compatibility as follows. 
We assume that agent~$\ag_i$ carries an independent $\text{Bernoulli}\left(p\right)^{\otimes \N}$ distributed random variable $U(\ag_i)$. 
There, $\otimes \N$ denotes the product space. Let the pool before the arrival of $\ag_i$ consist of the agents $\left\{\ag_{l_1}, \ldots, \ag_{l_K}\right\}$, where~$K$ is the size of the pool just before the arrival of $\ag_i$, and $1\leq l_1 < l_2 < \ldots < l_K < i$ are all integers. 
At arrival, $\ag_i$ forms an edge with $\ag_{l_j}$ if and only if $U_j(\ag_i)=1$, where $U_j(\ag_i)$ is the $j$th component of $U(\ag_i)$. 
So in particular, only the first $K$ many components of $U(\ag_i)$ are of interest.
Note that compatibility is symmetric and persistent, i.e., it is mutual and does not change once established.

While an agent~$\ag_i$ is in the pool, she can be matched with a compatible agent in the pool, and they leave the pool together. Otherwise, she leaves the pool when her random (maximum) \emph{\dtimename} $\sou_i$ lapses. 
We assume 
$\sou_1, \sou_2,\dots$
are i.i.d.\@ and independent of arrivals and compatibilities, with distribution 
$\mu$ on $[0,+\infty]$.
We call $\mu$ the \emph{\distname} and use
$\sou$ (without subscript) as a generic copy of the sojourn-time random variable.
We refer to the special case of $\sou$ where $\p(\sou = 1) = 1$, i.e., agents sojourn in the market for one time unit with probability~$1$, as a \emph{unit \dtimename}.
In this case, we have $\mu = \delta_1$, i.e., it is the degenerate probability measure with all mass at~$1$.

The last point, where an agent $\ag_i$ joining at time $t$ can be matched, is at time $t + \sou_i$, and we call an agent \emph{critical} at the moment when she reached her last time to be matched. 
The time period an agent spends in the pool is called her \emph{sojourn} and we say that an agent \emph{perishes} if she leaves the pool unmatched.

Agents are matched by certain \emph{matching algorithms} (or \emph{matching policies}), which decide when two compatible agents present in the current pool leave the market as a pair.
In principle, multiple pairs could be matched at the same time (as it is often done in batching algorithms, for instance for kidney exchange), and no agent can be part of more than one pair.
Note that algorithms have no other capability other than matching compatible agents.
For example, they cannot remove single unmatched agents to manipulate the waiting time, i.e., the time spent in the market.

Two important matching policies are the greedy algorithm and patient algorithm \citep{ALG20a}. 
On the one hand, the \emph{greedy algorithm} $\gdy$ is defined by the policy that every arriving agent is instantaneously matched to one of her compatible agents in the pool uniformly at random, or joins the pool if there exists no compatible agent. 
On the other hand, the \emph{patient algorithm} $\pat$ is defined to let agent~$\ag_i$ wait until she exhausts her maximum sojourn at time $t+\sou_i$.
Then $\pat$ checks $\pool_{(t+\sou_i)-}$, i.e., the pool at the beginning of this time step, and matches $\ag_i$ with a compatible agent uniformly at random, or lets her perish if no such agent exists.
Hence, before becoming critical, an agent only leaves the market when being matched to another critical agent.
Note that, because we have Poisson arrivals and independent \dtimenamepl, almost surely, for every time~$t$, it holds that at most one agent arrives or perishes at time~$t$.
Therefore, we disregard a case analysis for the possibility of multiple simultaneous exogenous events.

Some of our results hold for arbitrary algorithms.
In this case, they hold for any algorithm that has full knowledge of all current and future events, i.e., knowledge of future agent arrivals and all agents exact sojourn times.

\subsection{Objectives}

Our main objective is to minimize the loss of agents. 
Given an algorithm $\alg$ and an arrival rate $m$, we denote by $\alg(m,\Tmax)$ the set of agents matched until time $\Tmax$, i.e., $$\alg(m,\Tmax) := \{\ag_i\colon \ag_i \text{ matched until time } \Tmax\}\text.$$ 
Then, the \emph{loss of $\alg$ until time} $\Tmax$ is defined by
	$$\loss(m,\Tmax) := \frac{\E[|\arr_{[0,\Tmax]} \setminus (\alg(m,\Tmax) \cup \pool_\Tmax)|]}{m\Tmax}\text.$$

In this expression, the denominator is equal to the expected number of agents arriving in the time window $[0,\Tmax]$. 
Note that the loss of an algorithm also depends on the \dparaname~$d$ and the \distname $\sou$. 
As these will be fixed when we compute limits, we exclude them from our notation. 
However, we occasionally add the probability measure of the \distname as a subscript of probabilities or expectations to avoid ambiguity.

Following \citet{ALG20a}, we are interested in the limiting behavior for $m$ and $\Tmax$ tending to infinity. 
Therefore, we define the \emph{loss of $\alg$} as
$$\loss := \limsup_{m,\Tmax\to \infty}\loss(m,\Tmax)\text.$$

Note that considering the limiting behavior is mostly cosmetic. Error terms depending on $m$ and~$\Tmax$ decay quickly in our theoretical analysis, and our simulations show the robustness of our results for reasonably small arrival rates (see \Cref{fig:simu}).\footnote{For instance, the error terms in the analysis of the greedy algorithm in \Cref{lem:pool_greedy_upper_bound2}, \Cref{coro:greedy_pool_upper_bound}, and \Cref{lem:pool_greedy_upper_bound} are $m^{-9}$.}

Our next objective is the waiting time of agents, i.e., the time that agents have to spend in the pool.
For a given algorithm $\alg$, the 
\emph{average waiting time} until time $\Tmax$ is defined by
\begin{align*}
	\waitobj(m,\Tmax)
    := \frac 1{m\Tmax}\E\left[\sum_{i\ge 1}\int_{0}^{\Tmax} \mathbf{1}_{\ag_i \in Z_s} ds\right]\text.
\end{align*}

There, $\mathbf{1}$ denotes the indicator function.
Also, as in the definition of the loss, we suppress the dependence from the distribution of \dtimenamepl.
In other words, $\waitobj(m,\Tmax)$ is the sum over all agents of the expected time spent in the pool until time $\Tmax$ divided by $m\Tmax$, which is the expected number of agents arriving until time $\Tmax$.
Note that the integral implicitly depends on $m$ because the arrival rate influences the set of agents that has arrived until time $\Tmax$.

\section{Analysis of the Greedy Algorithm}\label{sec:gdyanalysis}

In this section, we discuss our results for the greedy algorithm.

\subsection{Loss Guarantee of the Greedy Algorithm}\label{sec:lossguarantee}
We start with our first objective, namely minimizing loss.
We will prove the following theorem.

\begin{restatable}{theorem}{epsgreedy}\label{thm:epslowerbound}
	Let $\eps>0$ and assume that the
    \dtimename $\sou$ satisfies $\p(\sou < \eps) = 0$.
    Then, for $d\geq 2$, we have
	\begin{equation*}
	\loss[\gdy] \leq e^{-\frac{\eps d}{2\log(2)}}\text.
	\end{equation*}
\end{restatable}

We remark that the theorem also holds 
for ex ante heterogeneous agents whose \dtimenamepl have different distributions but individually satisfy the respective assumptions.
In other words, we can replace the assumption $\p(\sou < \eps) = 0$ by $\p(\sou_i < \eps) = 0$ for all $i\ge 1$, assuming arbitrary distributions satisfying this assumption.\footnote{The same is true for all loss bounds in \Cref{sec:inevitableloss}.}

Towards a proof, we will first argue that the assumption in \Cref{thm:epslowerbound} of an arbitrary lower bound on the minimum \dtimename can be reduced to the special case of a bound of one time unit. 
This relies on a scaling invariance that connects the loss and \dtimename. 
To state the result, we need to express the loss with respect to the parameters that we usually suppress. 
Therefore, given a probability measure $\hat \mu$ and a \dparaname $\hat d$, let $\loss(\hat \mu,\hat d)$ denote the loss of $\alg$ with respect to a \dtimename with distribution $\hat \mu$ and \dparaname $\hat d$.

\begin{restatable}{lemma}{timechange}\label{lem:timechange}
	Let $\dep$ be a probability measure on $[0,+\infty]$, let $c>0$, and let $\nu$ be the probability measure defined by
    $\nu\left( \left[ a,b \right] \right) = \dep\left( \left[ ca,cb \right] \right)$ for all $a,b\in [0,+\infty]$.\footnote{While the specification of $\nu$ on intervals is sufficient, one can also define it by $\nu(B):=\dep(cB)$ for all Borel sets $B\subseteq[0,\infty]$, 
    i.e., $\nu$ is the pushforward of $\dep$ by $x\mapsto x/c$.}
    Then,
	\begin{align*}
	\loss[\gdy]\left(\nu,d \right) = \loss[\gdy]\left(\dep,cd \right)\text.
	\end{align*}
\end{restatable}

\begin{proof}
	The claim follows by a change of time. 
    If we consider the model with \dtimename distributed according to $\nu$ and speed up time by a factor of $c$, then agents arrive at Poisson rate $cm$ and the sped up \dtimename $\sou$ 
    satisfies
	\begin{align*}
	\p\left( \sou \leq a \right) = \nu\left( \left[ 0, a/c \right] \right) = \dep\left( \left[ 0, a \right] \right) \text.
	\end{align*}
	Thus, $\sou$ is distributed according to $\dep$. 
    However, two agents are still compatible with a probability of
	$\frac{d}{m} = \frac{cd}{cm}$.
	So, effectively we see a model with \dtimenamepl distributed according to $\dep$ and \dparaname $cd$.
\end{proof}

We will use this lemma to argue that it suffices to prove the following weaker variant of \Cref{thm:epslowerbound}.

\begin{restatable}{theorem}{gdyLossUB}\label{thm:greedy_loss_upper_bound}
	Assume that the
    \dtimename $\sou$ satisfies $\p(\sou < 1) = 0$.
    Then, for $d\geq 2$, it holds that
	\begin{equation}\label{eq:unitlossbound}
	\loss[\gdy] \leq e^{-\frac{d}{2\log(2)}}\text.
	\end{equation}
\end{restatable}

Indeed, to transition from a guaranteed sojourn time of~$1$ to a guaranteed sojourn time of $\eps$, we can scale the distribution of $\sou$ by $\eps$.
By \Cref{lem:timechange}, this leads to loss according to the \dparaname $\eps d$.
Hence, we obtain \Cref{thm:epslowerbound}.
In the remaining subsection, we discuss the proof of \Cref{thm:greedy_loss_upper_bound}.
We start with an outline of the proof in \Cref{sect:gdysketch}, and then go over key technical steps.
We defer the technical proofs of some steps to \Cref{app:greedyloss}.

\subsubsection{Outline of the Proof of \Cref{thm:greedy_loss_upper_bound}\label{sect:gdysketch}}

The main problem in the analysis of any other model than the one with exponentially distributed \dtimenamepl is the lack of Markovianity of the pool size. 
Knowing at which exact time points the agents in the current pool arrived, already reveals some information at which time points in the future we expect more or less agents to get critical and perish.
So the knowledge of the entire history $\left(\psize_t\right)_{t \leq t_0}$ typically reveals more information than the pool size $\psize_{t_0}$, which shows that $\left(\psize_t\right)_{t \geq 0}$ is not necessarily Markovian.
Hence, we cannot obtain bounds on the pool size merely by analyzing its stationary distribution.
Instead, we perform a direct analysis of the pool size which will be crucial to investigate the loss of the greedy algorithm.

Determining an upper bound for the pool size follows from the following idea: Assume that the pool has a size of $Cm$ for some constant $C>\frac{\log(2)}{d}$ and an agent arrives. 
On the one hand, there is a probability of
\begin{equation}\label{eq:simplebound}
\left( 1-\frac{d}{m} \right)^{Cm} \leq e^{-Cd} < \frac{1}{2} 
\end{equation}
that the new agent joins the pool causing the pool size to increase by $1$.\footnote{ 
The first inequality of \eqref{eq:simplebound} is a standard inequality. For completeness, it is proved in \Cref{lem:exp-estimate}.}
On the other hand, there is a probability of strictly more than $\frac{1}{2}$ that the agent matches and thus the pool size decreases by $1$. Consequently, for a sufficiently large pool size, the arrival of new agents decreases the pool size in average. Besides the arrival of new agents, existing agents in the pool may also perish and cause a further decrease of the pool size. 
This analysis shows that we do not expect the pool size to be much larger than $\frac{\log(2)}{d} m$ for most of the time. 
This key step is shown in  \Cref{coro:greedy_pool_upper_bound} and \Cref{lem:pool_greedy_upper_bound}. 
An important technique towards \Cref{coro:greedy_pool_upper_bound} is to compare an arbitrary \dtimename with the case when agents do not perish at all, i.e., with the case of an infinite \dtimename. For this, we will create a coupling of the respective markets.

Given a uniform bound on the pool size, we can also bound the loss under various assumptions on the \dtimename. 
For instance, assume unit \dtimenamepl. 
Further, suppose that a new agent, called $\ag$, enters the pool at time $t$ without being matched right away and that the pool size is bounded from above by $\frac{\log(2)}{d}m$ for all times in $\left[ t,t+1 \right]$.
This agent will stay for (at most) one time unit.
Moreover, we will show that each agent arriving during this time has a probability of 
$\Omega\left(\frac{d}{m}\right)$ to match with $\ag$, i.e., with this probability she is compatible with $\ag$ and $\ag$ is chosen as the uniformly random partner. 
Hence, the probability that $\ag$ perishes without being matched by another arriving agent crucially depends on the number of agents arriving during her sojourn.
For unit sojourn times, these are about $m$ agents.
Consequently, the probability of perishing is bounded by an expression of the form
$\left( 1- \Theta\left(\frac{d}{m}\right) \right)^{m} = e^{-\Theta(d)}$,
which explains the exponentially small loss under a constant \dtimename.

\subsubsection{Comparison of Sojourn Distributions}\label{subsec:greedy_compare}
In this section, we compare the pool sizes under the greedy algorithm for different distributions of \dtimenamepl and under different initial conditions.

In the first lemma, we use a coupling method to compare the evolution of the pool size under any \dtimename with the case of an infinite \dtimename.
Therefore, let $\sou_\infty$ be given by $\p(\sou_{\infty} = +\infty) = 1$, i.e., this is the random variable representing infinite sojourn.
Given a time~$t$, denote by $\pool_t^\infty$ and $\psize_t^\infty$ the pool and the pool size for \dtimenamepl according to $\sou_\infty$, respectively. 
Note that our coupling lemma leaves a lot of freedom for the starting pool.
For instance, it does not require empty pools in the beginning.

\begin{restatable}{lemma}{coupling}\label{lem:coupling1}
    Let $\sou$ be an arbitrarily distributed \dtimename.
    Consider the pools $\pool_t$ and $\pool_t^\infty$ at time $t$ for the greedy algorithm with 
    \dtimenamepl $\sou$ and $\sou_\infty$, respectively,
    and let $\pool_0$ and $\pool_0^\infty$ be two starting pools with $\psize_0\leq \psize_0^\infty$. 
    Then, there exists a coupling of $\pool_t$ and $\pool_t^\infty$ such that
	\begin{align}\label{eq:coupling1}
	\psize_t \leq \psize_t^\infty + 1
	\end{align}
	for all $t\geq 0$.
\end{restatable}

In particular, \Cref{lem:coupling1} implies that when running the greedy algorithm without perishing (infinite \dtimename) and one starts with two different initial pool sizes, then the process that started with the lower pool size can surpass the other one by at most~$1$. 

A key feature of the distribution~$\sou_\infty$ is that the pool size \emph{is} Markovian under this measure, as the evolution of the pool size depends only on the current size of the pool and the arriving agents.
The transition rates $(r_{k \to k+1})_{k\in \N_{\ge 0}}$ and $(r_{k \to k-1})_{k\in \N_{\ge 1}}$ are given by
\begin{align}\label{eq:transition-rates}
r_{k \to k+1} = m \left( 1- \frac{d}{m} \right)^{k} \ \text, \ r_{k\to k-1} = m \left[ 1- \left( 1- \frac{d}{m} \right)^{k} \right] \text.
\end{align}
Interestingly, we have that $r_{0 \to 1} = m$ and, for all $k\ge 1$, that
$r_{k \to k+1} + r_{k \to k-1} = m$.
Thus, if the pool size is larger than $0$, then the process waits an exponential time with expectation $1/m$ and then jumps to $k+1$ or $k-1$ with corresponding probabilities. 
At time~$0$, the process also waits for an exponential time with expectation $1/m$ and then jumps to $1$. 
In particular, the transition probabilities $\left(p(k,l)\right)_{k,l\in \N_{\ge 0}}$ of the underlying discrete-time Markov chain are simply given by
\begin{align}\label{eq:transition probs}
p(k , k+1) =  \left( 1- \frac{d}{m} \right)^{k} \text{ and } p(k , k-1) = 1- \left( 1- \frac{d}{m} \right)^{k}
\end{align}
for $k\in \N_{\ge 1}$, $p(0,1) = 1$, and all other probabilities are $0$. 
These transition probabilities give rise to an irreducible and reversible Markov chain (every Markov chain with state space $\N_{\ge 0}$ making jumps between nearest neighbors only is reversible). We claim that there exists a stationary distribution $\rho$ for this Markov chain. Indeed, one can construct a stationary measure with the aid of the measure~$\tilde \rho$ defined by $\tilde{\rho}(0)=1$ and $\tilde{\rho}(k+1)= \rho(k)\frac{p(k,k+1)}{p(k+1,k)} = \prod_{i=0}^k \frac{p(i,i+1)}{p(i+1,i)}$. 
This is {\sl not} a probability measure, yet. However, the measure~$\tilde{\rho}$ can be normalized as $\sum_{k=0}^\infty \prod_{i=0}^k \frac{p(i,i+1)}{p(i+1,i)} < \infty$, which follows because $p(i,i+1)<1/3$ for all large enough $i$.\footnote{In fact, choose $i_0$ with $(1-\frac{d}{m})^{i_0}\le \frac{1}{3}$. Then for $i\ge i_0$, it holds that
$\frac{p(i,i+1)}{p(i+1,i)}=\frac{(1-\frac{d}{m})^{i}}{1-(1-\frac{d}{m})^{i+1}}\le \frac{1/3}{2/3}=\frac12\text.$
Hence, $\sum_{k=0}^\infty\prod_{i=0}^k\frac{p(i,i+1)}{p(i+1,i)}<\infty$ by comparison with a geometric series.
}

The existence of a stationary distribution tells us that the pool size does not go off to infinity as time grows. The balance equation of the stationary distribution for $j\in \N_{\ge 0}$ reads
\begin{align}\label{eq:balance_equations}
\rho(j) p(j,j+1) = \rho(j+1) p(j+1,j)\text,
\end{align}
and we will give bounds on the stationary distribution in \Cref{lem:pool_greedy_upper_bound2} below.

\subsubsection{Pool Size Bounds for the Greedy Algorithm}\label{subsec:Greedy_pool size_bound}

In this section, we provide bounds for the pool size of the greedy algorithm. The key idea is to distinguish different ranges of the pool size constraint by carefully chosen constants. For this, we define the constants $C_1 = C_1(m)$, $C_2 = C_2(m)$, and $C_3 = C_3(m)$ by
\begin{align*}
C_1 = 1+ \frac{10}{\log(2)\log(m)} \text, \ 
C_2 = C_1+2 \frac{d \log(m)^2}{m \log(2)} \text{, and }
C_3 = C_1 + 4 \frac{d \log(m)^2}{m \log(2)} \text.
\end{align*}

First, we show how to control the stationary measure under infinite \dtimenamepl. 
As a consequence, by Lemma \ref{lem:coupling1}, we obtain bounds for the pool size under an arbitrary \distname. 

\begin{restatable}{lemma}{greedyUbPool}\label{lem:pool_greedy_upper_bound2}
	Let $\rho$ be the stationary measure of the pool size without perishing. Then,
	\begin{align*}
	\rho \left( \left(\frac{C_1 \log(2) m}{d} + \frac 32\log(m)^2 ,+\infty \right) \cap \N \right) \leq m^{-9}
	\end{align*}
	for $m$ large enough.
\end{restatable}

Combining the bounds under infinite sojourns with our coupling result, we can bound the pool size of the greedy algorithm under an arbitrary \dtimename.

\begin{restatable}{corollary}{corGdyPoolUb}\label{coro:greedy_pool_upper_bound}
	Assume that the \dtimename $\sou$ is according to an arbitrary distribution $\mu$.
    Then, it holds for all times~$t$ and arrival rates~$m$ large enough that
	\begin{align*}
	\p_\dep \left( \psize_t > \frac{C_2 \log(2)}{d}m \right) \leq m^{-9}.
	\end{align*}
\end{restatable}

\begin{proof}
	Assume that $\pool_0$ is the empty pool and $\pool_0^\infty$ has size distributed according to the stationary distribution~$\rho$. Then $\psize_0 \leq \psize_0^\infty$, so we can couple the processes according to \Cref{lem:coupling1} such that $\psize_t \leq \psize_t^\infty+1$. The pool size $\psize_t^\infty$ is still distributed according to $\rho$ by stationarity, so we get that for large enough $m$
	\begin{align*}
	\p_\dep\left( \psize_t > \frac{C_2 \log(2)}{d}m \right) \leq 
	\p_{\dep_\infty}\left( \psize_t^\infty > \frac{C_1 \log(2)}{d}m + \frac 32 \log(m)^2 \right) \leq m^{-9}\text.
	\end{align*}
	
	There, the first inequality follows from \Cref{lem:coupling1} and the second inequality from \Cref{lem:pool_greedy_upper_bound2}.
\end{proof}

Our second key insight is that if the pool size is sufficiently small, then the probability of reaching a much larger pool size within one time unit is very low.

\begin{restatable}{lemma}{poolgdyub}\label{lem:pool_greedy_upper_bound}
	Assume that the \dtimename $\sou$ is according to an arbitrary distribution $\mu$.
	Further, assume that $\pool$ is a pool of size at most $C_2 \frac{\log(2)}{d} m$. 
    Then, for large enough $m$, it holds that
	\begin{align*}
		& \p_{\dep}\left( \psize_t \leq C_3 \frac{\log(2)m}{d} \text{ for all } t \leq 1 \  \big| \ \pool_0=\pool  \right) \geq 1 - m^{-9}\text.
	\end{align*}
\end{restatable}

\subsubsection{Loss Bound for the Greedy Algorithm}\label{subsec:loss_bound_greedy}

Before we leverage the bounds on the pool size to obtain bounds on the loss, we provide a useful lemma that lets us express the loss as an integral over the probability of perishing. 
This lemma holds for arbitrary matching algorithms and we will apply it again for the analysis of the patient algorithm. 
It follows from an application of Mecke's equation \citep[see, e.g.,][Theorem~4.1]{LaPe17a}. 

\begin{restatable}{lemma}{auxloss}\label{lem:aux-loss}
	Let $\alg$ be a matching algorithm and assume that the \dtimename $\sou$ is distributed according to an arbitrary distribution $\dep$.
    Then, 
	\begin{align*}
		\loss(m,\Tmax) = \frac 1\Tmax \int_{0}^{\Tmax} \p_{\dep}\left(\text{an agent arriving at time $t$ perishes before time $\Tmax$} \right) dt \text.
	\end{align*}
    
	Further, if the sojourn time is almost surely finite, i.e., if $\mu \left( \{+\infty\} \right) = 0$, then
	\begin{align*}
	\loss = \limsup_{m,\Tmax\to \infty} \frac 1{\Tmax} \int_{0}^{\Tmax} \p_{\dep}\left(\text{an agent arriving at time $t$ perishes} \right) dt \text.
	\end{align*}
\end{restatable}

Finally, we can apply our insights to bound the loss of the greedy algorithm and prove \Cref{thm:greedy_loss_upper_bound}.

\gdyLossUB*

\begin{proof}
	By \Cref{lem:aux-loss}, we have 
	\begin{align}\label{eq:integral greedy upper bound}
	\loss[\gdy] = \limsup_{m,\Tmax\to \infty} \frac 1 {\Tmax}
	\int_0^{\Tmax} \p_\dep \left( \text{an agent arriving at time $t$ perishes} \right) dt\text.
	\end{align}
    
	We now define a certain ``good'' event $\mathcal{G}_t$ on which we will show that the probability that an agent perishes is relatively low. 
    Let
	\begin{align*}
	\mathcal{G}_t := 
	&
	\left\{ \psize_t \leq  \frac{C_2 \log(2)}{d}m \right\}
	\cap \left\{ \psize_s \leq  \frac{C_3 \log(2)}{d}m \forall s\in \left[t,t+1\right] \right\}
	\cap \left\{ \numbarr_{\left[t,t+1\right]} \geq m- 30\log(m)m^{1/2}\right\} \text.
	\end{align*}
    
	We first want to bound the probability of the complement~$\mathcal{G}_t^c$. By a union bound, we get
	\begin{align*}
	\p_\dep \left( \mathcal{G}_t^c \right) \leq &
	\p_\dep \left( \psize_t >  \frac{C_2 \log(2)}{d}m  \right) +\p_\dep \left( \numbarr_{\left[t,t+1\right]} < m- 30\log(m)m^{1/2} \right) \\
	&
	+\p_\dep \left( \left\{ \psize_s \leq  \frac{C_3 \log(2)}{d}m \text{ for all } s\in \left[t,t+1\right] \right\}^c \cap \left\{ \psize_t \leq  \frac{C_2 \log(2)}{d}m \right\}  \right) \ .
	\end{align*}
    
	By \Cref{coro:greedy_pool_upper_bound}, the probability of the first summand is bounded by $m^{-9}$. 
    For the second summand, note that $\numbarr_{\left[t,t+1\right]}$ is a Poisson random variable with parameter $m$ and thus we can use a Chernoff bound, see \Cref{lem:Chernoff}, obtaining
	\begin{align}
		\p_\dep \left( \numbarr_{\left[t,t+1\right]} < m- 30\log(m)m^{1/2} \right) 
		&
		= \p_\dep \left( \numbarr_{\left[t,t+1\right]} < m \left( 1- \frac{30\log(m)}{m^{1/2}} \right) \right)\notag\\
		&
		\leq e^{- \frac{m}{3} \left(\frac{30\log(m)}{m^{1/2}}\right)^2} \leq m^{-10}\text. \label{eq:poolovertimebound}
	\end{align}
    
	There, the last two inequalities hold for $m$ large enough. 
    The last summand is bounded by
	\begin{equation*}
	\p_\dep \left( \psize_s >  \frac{C_3 \log(2)}{d}m \text{ for some } s\in \left[t,t+1\right]  \ \big| \  \psize_t \leq  \frac{C_2 \log(2)}{d}m   \right) \leq m^{-9}
	\end{equation*}
	by \Cref{lem:pool_greedy_upper_bound}. 
    Together, this shows that $\p_\dep \left( \mathcal{G}_t^c \right) \leq m^{-8}$ for large enough $m$.
	
	Now assume that the event $\mathcal{G}_t$ holds and an agent $\ag$ arrives at time $t$. 
	We would like to derive the probability that $\ag$ does neither get matched directly nor between time $t$ and $t+1$. 
    Whenever a new agent $\agtwo$ arrives at time $s\in(t,t+1)$ and at a pool of size $\psize_s$, then the probability that $\agtwo$ matches with $\ag$ equals
	\begin{align*}
	\frac{1}{\psize_s} \left[ 1-\left(1-\frac{d}{m}\right)^{\psize_s} \right]\text.
	\end{align*}
    
	Indeed, there is a probability of $ 1-\left(1- \frac{d}{m} \right)^{\psize_s}$ that the agent $a$ matches at arrival, and all vertices in the pool have the same probability of getting matched. 
    The function $x \mapsto \frac{1}{x} \left[ 1-\left(1-\frac{d}{m}\right)^{x} \right]$ is decreasing for $x>0$ which can be seen as follows. Write $b=1-\frac{d}{m}$. The derivative of this function is given by $\frac{b^x-xb^x \ln(b)-1}{x^2}$. As $1-x\ln(b) \leq e^{-x \ln(b)} = b^{-x}$
	we already have that
	\begin{align*}
		\frac{b^x-xb^x \ln(b)-1}{x^2} = \frac{b^{x}}{x^2} \left( 1-x \ln(b) - b^{-x}  \right) \leq  0\text,
	\end{align*}
	which shows that the function is decreasing. 
    
    Using our upper bound on the pool size on the event $\mathcal{G}_t$, we obtain
	\begin{align*}
		\frac{1}{\psize_s} \left[ 1-\left(1-\frac{d}{m}\right)^{\psize_s} \right] \geq \frac{d}{C_3 \ln(2) m} \left[ 1-\left(1-\frac{d}{m}\right)^{ \frac{C_3 \ln(2) m}{d} } \right] 
	\end{align*}
	for all $s\in \left[t,t+1\right]$. 
    
    As $C_3>1$ we have that
	\begin{align*}
		\left(1-\frac{d}{m}\right)^{ \frac{C_3 \ln(2) m}{d} } \leq \left(1-\frac{d}{m}\right)^{ \frac{ \ln(2) m}{d} } \overset{\text{\Cref{lem:exp-estimate}}}{\leq} \frac{1}{2}\text.
	\end{align*}
	
    Thus, the probability that agent $\agtwo$ matches with $\ag$ is at least
	\begin{align*}
		\frac{1}{\psize_s} \left[ 1-\left(1-\frac{d}{m}\right)^{\psize_s} \right] \geq \frac{d}{C_3 \ln(2) m} \left[ 1-\left(1-\frac{d}{m}\right)^{ \frac{C_3 \ln(2) m}{d} } \right] \geq \frac{d}{2 C_3 \ln(2) m}\text.
	\end{align*}

    Finally, consider the event $\mathcal{A}_t := \left\{ \text{an agent arriving at time $t$ perishes}  \right\}$. 
	Using $\numbarr_{\left[t,t+1\right]}\geq m-30\log(m)m^{1/2}$ and the assumption on the support of $\dep$, implying that an agent can only perish after a sojourn of at least one time unit, we can further bound the probability that the agent $a$ perishes from above by
	\begin{align}
	\p(\mathcal{A}_t \mid \mathcal{G}_t)\le \left( 1-\frac{d}{ 2 C_3 \log(2) m} \right)^{m-30\log(m)m^{1/2}} \underset{m\to \infty}{\longrightarrow} e^{-\frac{d}{2 \log(2)}}\text.\label{eq:singleaglossgood}
	\end{align}

    To extend this inequality without the restriction of $\mathcal{G}_t$, we compute 
	\begin{align}
	\p_\dep & \left( \mathcal{A}_t \right) = 
	\p_\dep \left( \mathcal{A}_t \mid \mathcal{G}_t \right) \p_\dep \left( \mathcal{G}_t \right)
	+ 
	\p_\dep \left( \mathcal{A}_t \mid \mathcal{G}_t^c \right) \p_\dep \left( \mathcal{G}_t^c \right)\notag\\
	&
	\leq \p_\dep \left( \mathcal{A}_t \mid \mathcal{G}_t \right) + \p_\dep \left( \mathcal{G}_t^c \right) 
	\leq 
	\p_\dep \left( \mathcal{A}_t \mid \mathcal{G}_t \right) + m^{-8}\text.\label{eq:singleagloss}
	\end{align}

    By \eqref{eq:singleaglossgood}, this can be bounded by $e^{-\frac{d}{2 \log(2)}}$ in the limit for $m$ tending to $\infty$. 
	Inserting this into \eqref{eq:integral greedy upper bound} finishes the proof.
\end{proof}

\subsection{Inevitable Loss}\label{sec:inevitableloss}

In the last section, we have seen that the loss of $\gdy$ is exponentially small if agents stay in the market for a guaranteed time period.
We will complement this by carving out reasons for a high loss of the greedy algorithm.
Our next lemma is a general statement about lower bounds for the loss.

\begin{restatable}{lemma}{greedylb}\label{lem:greedy_lb}
	Let $\eps,\delta >0$ and assume that the \dtimename $\sou$ has a probability distribution $\dep$ with $\dep\left( \left[ 0,\eps \right] \right) \geq \delta$. Then,
	\begin{align*}
	\loss[\gdy] \geq \frac{\delta}{2} e^{-\eps d}\text.
	\end{align*}
\end{restatable}

\begin{proof}
	Consider an agent~$\ag$ arriving at time~$t$ which is not matched upon arrival. 
    Note that compatible agents come at Poisson rate $m\cdot p = m\cdot \frac{d}{m} = d$. 
    As the maximum sojourn time and the arrival of compatible agents are independent, we get
	\begin{align*}
	& \p_\dep\left( \ag \text{ perishes} \mid \ag \text{ not matched upon arrival} \right) \\
	&
	\geq \p_\dep\left( \text{no compatible agent arrives in } \left[ t,t+\sou \right]  \mid \ag \text{ not matched upon arrival} \right) \\
	&
	\geq \p_\dep\left( \sou \leq \eps \right) \p_\dep\left( \text{no compatible agent arrives in } \left[ t,t+\eps \right] \right) \geq \delta e^{-\eps d}.
	\end{align*}
    
	Additionally, $\p_\dep\left( \ag \text{ not matched upon arrival} \right) \ge \frac 12$, because at least half of the agents $\ag$ are not matched upon arrival. Together,
	\begin{align*}
		 \p_\dep\left( \ag \text{ perishes} \right) = &\p_\dep\left( \ag \text{ perishes} \mid \ag \text{ not matched upon arrival} \right)\\  
		 &\cdot\p_\dep\left( \ag \text{ not matched upon arrival} \right) \ge \frac 12 \delta e^{-\eps d}\text.
	\end{align*}
    This proves the desired lower bound.
\end{proof}

An immediate consequence of this lemma is an inevitable inverse linear loss of $\gdy$ if the distribution of $\sou$ can be lower-bounded by any linear function in any neighborhood of~$0$.

\begin{restatable}{theorem}{lbnowaiting}\label{thm:highloss}
	Let $c>0$, $\eps_0 > 0$, and assume that the \dtimename is distributed according to a probability measure $\dep$ with $\dep\left( \left[0,\eps \right] \right) \geq c \eps$ for all $\eps \leq \eps_0$. Then, for~$d$ with $\frac 1d\le \eps_0$, it holds that
	\begin{align*}
	\loss[\gdy] \geq \frac{c}{6}\cdot \frac {1}{d}\text.
	\end{align*}
\end{restatable}

\begin{proof}
    Consider any $d$ with $\frac 1d \le \eps_0$ and $\eps = \frac{1}{d}$.
    By assumption, we have $\dep([0,\eps]) \ge c\eps$.
    Hence, \Cref{lem:greedy_lb} yields 
    \begin{equation*}
        \loss[\gdy] \geq \frac{c\eps}{2} e^{-\eps d} = \frac{c\eps}{2} e^{-1} = \frac c{2e}\cdot\frac 1d \ge \frac c6\cdot\frac 1d\text.
    \end{equation*}
    This proves the desired bound.
\end{proof}

Hence, the inferiority of the greedy algorithm  for exponentially distributed \dtimenamepl observed by \citet{ALG20a} is caused by a non-negligible chance of staying in the market for a very small time.
Importantly, it is not the mean of the \dtimename that leads to the different loss bounds in \Cref{thm:greedy_loss_upper_bound,thm:highloss} but the qualitative behavior of the distribution around~$0$.
Note that \Cref{thm:highloss} encompasses the exponential distribution.
Assume that $\sou$ is exponentially distributed with mean~$1$, i.e., $\p(\sou \le x) = 1 - e^{-x}$ for all $x\ge 0$.
Note that for all $x\ge 0$, it holds that $1- e^{-x} \ge \frac x{1+x}$ (see \Cref{lem:exp-snd-estimate}).
Moreover, we have $\frac x{1+x} \ge cx$ whenever $x\le \frac 1c - 1$.
Hence, the precondition of \Cref{thm:highloss} is, for example, satisfied with $\eps_0 = 1$ and $c = \frac 12$.
Hence, for any \dparaname $d\ge 1$, and we then get a loss of $\frac 1{12d}$.

Moreover, another simple consequence of \Cref{lem:greedy_lb} is a general exponential lower bound for the loss of $\gdy$ under constant \dtimenamepl and, therefore, in particular, unit \dtimenamepl.

\begin{restatable}{proposition}{greedylbconst}\label{cor:greedylbconst}
	Let $k > 0$ and assume that the \dtimename $\sou$ is given by $\p(\sou = k) = 1$. 
    Then,
	\begin{align*}
	\loss[\gdy] \geq \frac{1}{2} e^{- kd}\text.
	\end{align*}
\end{restatable}

\begin{proof}
    This result follows from \Cref{lem:greedy_lb} by setting $\eps = k$ and $\delta = 1$.
\end{proof}

Hence, for the \dtimenamepl of \Cref{cor:greedylbconst} an exponentially small loss as derived in \Cref{thm:epslowerbound} is the best the greedy algorithm can achieve.
Up to a constant in the exponent, $\gdy$ cannot guarantee better outputs.

However, a much stronger statement is true, as we show next. 
An exponential loss is not only inevitable for $\gdy$ but for any algorithm, as long as the expected sojourn time is finite. 
Note that this result even holds for algorithms that have knowledge of all future agents and their compatibilities. 

\begin{restatable}{theorem}{lossgenlb}\label{thm:lossgenlb}
	Let $E := \E[\sou]$ be the expected \dtimename.\footnote{In other words, $E = \int x \dep(dx)$.}
    Then, for any algorithm $\alg$, it holds that
	\begin{align*}
		\loss[\alg] \geq e^{-2Ed}\text.
	\end{align*}
\end{restatable}

\begin{proof}
    We start with a technical observation.
	Let $X$ be any non-negative random variable with distribution $\dep$. 
    It is a well-known fact (tail integral formula) that
	\begin{align*}
		\int_{0}^{\infty} \dep\left(\left[s,\infty\right]\right) ds
		=  \E \left[X\right]\text.
	\end{align*}
	
	We proceed with the main proof.
	For some time $t \in \R_{\geq 0}$, let $V_t$ be the set of agents whose maximum sojourn time is reached after time $t$. 
	So $V_t$ is the set of all agents that are possibly in the pool at time $t$. Mecke's equation implies that
    \begin{align}
        \notag \E \left[ |V_t| \right] & = \int_{0}^{t} m \p_{\mu} \left( \text{an agent arriving at time $s$ has sojourn time at least $t-s$} \right) ds
        \\
        & \label{eq:intermediary}
        =
        m \int_{0}^t \mu \left( \left[ t-s , \infty \right] \right) ds
        =
        m \int_{0}^t \mu \left( \left[ s , \infty \right] \right) ds
        \leq
        m \int_{0}^{\infty} \mu \left( \left[ s , \infty \right] \right) ds
        =
        m \E \left[X \right] = mE
    \end{align}

	Now, assume that an agent $\ag$ arrives at time $t$. 
    We want to estimate the probability that no agent in $V_t$ is compatible with $\ag$.
    By the law of total probability we have
	\begin{align}\label{eq:earlier arrivals}
		\notag &\p \left(\text{No agent in $V_t$ is compatible with $\ag$}\right)\\
		\notag & = \sum_{k=0}^{\infty} \p \left(|V_t|=k\right) \p \left(\text{no agent in $V_t$ is compatible with $\ag$}  \ \big| \ |V_t|=k\right)\\
		& =
		\sum_{k=0}^{\infty} \p \left(|V_t|=k\right) \left(1-\frac{d}{m}\right)^k
		=
		\E \left[ \left(1-\frac{d}{m}\right)^{|V_t|} \right] 
		\geq
		\left(1-\frac{d}{m}\right)^{\E \left[|V_t| \right]} 
		\overset{\eqref{eq:intermediary}}{\geq}
		 \left(1-\frac{d}{m}\right)^{mE}\text.
	\end{align}
    
	The first inequality holds by Jensen's inequality, as the function $x\mapsto \left(1-\frac{d}{m}\right)^x$ is convex. 
	The second inequality holds because this function is also decreasing.
	
	Next, we estimate the probability that no agent compatible with $\ag$ arrives during the sojourn time of $\ag$. 
	Let $X_{\ag}$ denote the maximum sojourn time of $\ag$, i.e., the random variable $X_{\ag}$ is distributed according to $\dep$. 
	The next idea is to condition on the event $\{X_{\ag}=s\}$ for some positive real $s$. Then, the law of total probability and Jensen's inequality imply that
	\begin{align*}
		&\p \left( \text{no agent arriving in $\left[t,t+s\right]$ is compatible with $\ag$} | X_{\ag}=s \right)\\
		=
		&\sum_{k=0}^{\infty} \p\left(a_{\left[t,t+s\right]}=k | X_{\ag}=s\right) \p \left( \text{no agent arr. in $\left[t,t+s\right]$ is comp. with $\ag$} \ | \ X_{\ag}=s ,  a_{\left[t,t+s\right]}=k \right)\\
		=
		&\sum_{k=0}^{\infty} \p\left(a_{\left[t,t+s\right]}=k\right)\left(1-\frac{d}{m}\right)^k =
		\E\left[\left(1-\frac{d}{m}\right)^{a_{\left[t,t+s\right]}}\right] \geq
		\left(1-\frac{d}{m}\right)^{\E\left[{a_{\left[t,t+s\right]}}\right]}
		= \left(1-\frac{d}{m}\right)^{ms}\text.
	\end{align*}
    By the tower rule and the law of total probability we thus get
	\begin{align}\label{eq:later arrivals}
		\notag &\p \left( \text{no agent arriving in $\left[t,t+X\right]$ is compatible with $\ag$} \right)\\
		\notag &
		=
		\E \left[ \p \left( \text{no agent arriving in $\left[t,t+X\right]$ is compatible with $\ag$} | X \right) \right]\\
		&
		\geq 
		\E \left[  \left(1-\frac{d}{m}\right)^{mX} \right] \geq
		 \left(1-\frac{d}{m}\right)^{\E \left[ mX\right]}
		 =
		  \left(1-\frac{d}{m}\right)^{mE}\text.
	\end{align}
	In the last inequality, we have applied Jensen's inequality once again. 
	
	An agent that arrives at time $t$, is compatible with no agent in $V_t$, and for whom no compatible agent arrives during her maximal sojourn time, cannot get matched under any algorithm. 
    For an agent arriving at time $t$, the compatibility to agents that arrived before time $t$ and the compatibility to agents that arrive after time $t$ is independent. 
    Hence, combining \eqref{eq:earlier arrivals} with \eqref{eq:later arrivals} gives that
	\begin{align*}
		\p \left(\ag \text{ perishes}\right) \geq \left(1-\frac{d}{m}\right)^{2mE}
	\end{align*} 
	for any agent $\ag$. 
	For $m\to \infty$, we obtain that $\p \left(v \text{ perishes}\right) \geq e^{-2dE}$, and as the time $t$ at which $\ag$ arrives was arbitrary we get by \Cref{lem:aux-loss} that
	\begin{align*}
		\loss[\alg] \geq e^{-2dE}
	\end{align*}
	for any algorithm $\alg$.
\end{proof}

For unit \dtimenamepl, \Cref{thm:lossgenlb} implies that,
for any algorithm $\alg$, it holds that
\begin{align*}
\loss[\alg] \geq e^{-2d}\text.
\end{align*}

In comparison, \Cref{thm:greedy_loss_upper_bound} bounds the loss of $\gdy$ as
\begin{equation*}
	\loss[\gdy] \leq e^{-\frac{d}{2\log(2)}}\text.
\end{equation*}

Hence, the constant in the exponential function of the loss generated by $\gdy$ for unit \dtimenamepl can be at most a factor of $4\ln(2)\approx 2.77$ worse compared to an optimal algorithm.

\subsection{Waiting Time under the Greedy Algorithm}\label{sec:waiting time}

Our second goal is the minimization of waiting time.
For the greedy algorithm, since almost all agents get matched, approximately half of the agents get matched at arrival.
However, the waiting time of agents entering the pool is unclear. 
We will show that, in the greedy algorithm under an arbitrary \dtimename, $\waitobj[\gdy](m,\Tmax) \in \Theta \left(\frac{1}d \right)$. 

Interestingly, the average waiting time can be expressed by an integral over the expected average pool size. 
This can be seen as a variant of Little's Law in our context.
\begin{restatable}{proposition}{waitVsLoss}\label{prop:wait-vs-loss}
	Assume that the \dtimename is distributed according to an arbitrary probability measure $\dep$. Then, for any algorithm $\alg$, it holds that
	\begin{align*}
		\waitobj[\alg](m,\Tmax) = \frac 1{m\Tmax}\int_0^{\Tmax} \E_{\dep}\left[ \psize_s \right] ds \text. 
	\end{align*}
\end{restatable}

\begin{proof}
    By Fubini's Theorem, we have
	\begin{align*}
		 \sum_{i\ge 1} \int_{0}^{T} \mathbf{1}_{v_i \in Z_s} ds 
		=
		\int_{0}^{T} \sum_{i \ge 1} \mathbf{1}_{v_i \in Z_s} ds 
		=
		\int_{0}^{T}\psize_s ds\text.
	\end{align*}
	
    There, $\mathbf{1}$ denotes the indicator function.
    Hence,
    \begin{equation*}
        \waitobj(m,\Tmax)
        = \frac 1{m\Tmax}\E\left[\sum_{i\ge 1}\int_{0}^{\Tmax} \mathbf{1}_{v_i \in Z_s} ds\right] = \frac 1{m\Tmax}\E\left[\int_{0}^{T}\psize_s ds\right]\text.
    \end{equation*}
    
    Another application of Fubini's theorem yields 
	\begin{align*}
		\waitobj(m,T) = \frac 1{m\Tmax}\int_0^T \E_{\dep}\left[ \psize_s \right] ds \text.
	\end{align*}
    
    This proves the desired equation.
\end{proof}

We can use this relationship to give an upper bound on the waiting time.

\begin{restatable}{theorem}{aggWaiting}\label{lem:agg_waiting}
	Assume that the \dtimename is distributed according to an arbitrary probability measure $\dep$. Then, for large enough $m$, under the greedy algorithm, it holds that
	\begin{align*}
		\waitobj[\gdy](m,\Tmax) \leq \frac 65\cdot\frac{1}{d}\text.
	\end{align*}
\end{restatable}

\begin{proof}
	Having in mind \Cref{prop:wait-vs-loss}, we would like to get an upper bound on $\E_{\dep}\left[ \psize_s \right]$. 
    The key idea for this is to consider the Markov chain $\left(\pool_t^\infty\right)_{t\geq 0}$ with infinite \dtimename started at stationarity and to apply \Cref{lem:coupling1}, our coupling lemma.

    Recall that we obtained the transition probabilities  with respect to $\dep_\infty$ in \Cref{eq:transition probs} in \Cref{subsec:greedy_compare} as $p(k,k+1) = \left(1-\frac{d}{m}\right)^k$.
	By \Cref{lem:exp-estimate}, we have $\left(1-\frac{d}{m}\right)^k \leq 1/3$ for $k\geq \frac{\log(3)m}{d}$. 
    Thus, the stationary distribution $\rho$ satisfies $\rho(k+1) = \frac{p(k,k+1)}{p(k+1,k)} \rho(k) \leq \frac{1}{2}\rho(k)$ for these $k$. 
	However, this already implies that when $\psize_0^\infty$ is distributed according to the stationary distribution, then
	\begin{align*}
		& \E_{\dep_\infty}\left[\psize_0^\infty\right] = \sum_{k=1}^{\infty} \rho(k)k
		=
		\sum_{k=1}^{\left\lceil \frac{\log(3)m}{d}\right\rceil + 10} \rho(k)k 
		+
		\sum_{\left\lceil \frac{\log(3)m}{d}\right\rceil + 11}^\infty \rho(k) k\\
		&
		\leq
		\left\lceil \frac{\log(3)m}{d}\right\rceil + 10 
		+ \sum_{k=11}^{\infty}2^{-k} \left(\left\lceil \frac{\log(3)m}{d}\right\rceil + k\right)\\
		&
		\leq  \left\lceil  \frac{\log(3)m}{d}\right\rceil + 10  +  2^{-10} \left\lceil \frac{\log(3)m}{d}\right\rceil + 1
		\leq 1.01 \left\lceil \frac{\log(3)m}{d}\right\rceil + 11\text.
	\end{align*}
Thus, we can apply \Cref{lem:coupling1} to obtain
\begin{align*}
	m\Tmax\cdot \waitobj[\gdy](m,T) &= \int_0^T \E_{\dep}\left[ \psize_s \right] ds \\
	&\leq \int_0^T \E_{\dep_\infty}\left[ \psize_s^\infty + 1 \right] ds \\
    &\leq T\left( 1.01 \left\lceil \frac{\log(3)m}{d}\right\rceil + 12 \right) \leq \frac{6mT}{5d}\text,
\end{align*}
where the last inequality holds for $m$ large enough.
\end{proof}

We also prove a lower bound on the average waiting time.
Even if an algorithm has access to the whole information of incoming agents, compatibilities, and the realization of maximum sojourn times, it has a waiting time that is only a constant factor apart from the waiting time of the greedy algorithm. 
Note that it is important for this result that algorithms are not allowed to remove unmatched agents from the market.
This would allow them to obtain an arbitrarily small waiting time at the expense of having a high proportion of unmatched agents.
Moreover, this lower bound needs two minor additional assumptions. 
First, the density of the sojourn distribution around a sojourn of~$0$ is not allowed to be too high.
Note that this does \emph{not} exclude the exponential distribution: the conditions of \Cref{prop:waitlb} hold for the exponential distribution with mean~$1$ whenever $d\ge 5$.
Moreover, we need a time horizon of at least $\frac 3d$.
As $d\ge 1$, this is satisfied whenever $T\ge 1$, i.e., the market runs for at least one time unit.

\begin{restatable}{proposition}{waitlb}\label{prop:waitlb}
	Assume that the \dtimename is distributed according to a probability measure $\dep$ with $\dep \left( \left[\frac{1}{2d}, \infty\right] \right) > \frac{9}{10}$. 
    Further, assume that $\Tmax\ge \frac 3d$. 
    Then, under any algorithm $\alg$, it holds that
	\begin{align*}
		\waitobj[\alg](m,\Tmax) \geq \frac 1{20}\cdot\frac{1}{d}\text.
	\end{align*}
\end{restatable}

\begin{proof}
	We say that an agent arriving at time $t \in \left[\frac{1}{2d}, T -\frac{1}{2d}\right]$ is {\sl blocked} if no compatible agent arrives in the interval $\left[t-\frac{1}{2d},t+\frac{1}{2d}\right]$, and the maximum sojourn time of this agent is at least $\frac{1}{2d}$. So in particular the agent will still be in the market at time $t+\frac{1}{2d}$. For some agent arriving at time $t$, compatible agents arrive at rate $d$. So the probability that there arrives no compatible agent in the time interval $\left[t-\frac{1}{2d},t+\frac{1}{2d}\right]$ is given by $\exp \left( -d \frac{1}{d} \right) = e^{-1}$. 
	
	As the arrival process of compatible agents is independent of the maximum sojourn time, we get that for an agent $\ag$ arriving at time $t \in \left[\frac{1}{2d}, T -\frac{1}{2d}\right]$
	\begin{align*}
		\p \left( \ag \text{ is blocked} \right) = \dep \left(\left[\frac{1}{2d},\infty\right]\right) e^{-1} \geq \frac{0.9}{e} 
	\end{align*}
	by our assumption on $\dep$. 	
	Let $B(T)$ be the number of blocked agents arriving in the interval $\left[\frac{1}{2d}, T -\frac{1}{2d}\right]$. 
	An application of Mecke's equation as in \Cref{lem:aux-loss}, yields that
	\begin{align*}
		\E \left[|B(T)|\right] \geq m\left(T-\frac{1}{d}\right)\frac{0.9}{e}\text.
	\end{align*}

    In the following, we will argue that under {\sl any} matching algorithm one, it holds that 
    \begin{equation*}
        m\Tmax\cdot \waitobj[\alg](m,\Tmax) \geq \frac{\E[|B(T)|]}{4d}\text.    
    \end{equation*}

	Agents can leave the market either alone when they perish or in pairs of two. 
    Assume that a blocked agent $\ag$ leaves the market alone, then it had a waiting time of at least $\frac{1}{2d}$. 
	If a blocked agent $\ag$ gets matched with an agent that arrived after $\ag$, then the blocked agent also had to wait for a time of at least $\frac{1}{2d}$. 
	If a blocked agent $\ag$ gets matched with an agent that arrived before $\ag$, then the agent with whom $\ag$ got matched had to have a waiting time of at least $\frac{1}{2d}$, as no agent compatible with $\ag$ arrived in the time window of length $\frac{1}{2d}$ before the arrival of $\ag$. This also leads to a higher total waiting time. However, we need to take care of the possibility of double counting here. 
	If two blocked agents are blocked and compatible with each other, then we must not accumulate a contribution of $\frac{1}{2d}$ to the total waiting time for each of them but only for the pair. 
	We can take care of this by a division by $2$. Hence, by monotonicity of expectations, it holds that
	\begin{align*}
		m\Tmax \cdot \waitobj[\alg](m,T) \geq \frac{1}{2d} \frac{\E[|B(T)|]}{2}\text. 
	\end{align*}
    
	Thus,
	\begin{align*}
		m\Tmax \cdot \waitobj[\alg](m,T) \geq \frac{1}{4d} \E \left[|B(T)|\right] \geq \frac{9}{40ed} m \left(T-\frac{1}{d}\right) \geq \frac{mT}{20d}\text.
	\end{align*}
	The last inequality holds, as $T\ge \frac 3d$ which is equivalent to $T-\frac{1}{d} \geq \frac{2T}{3}$,  and
	$\frac{9}{40e}\frac{2}{3} \geq \frac{1}{20}$.
\end{proof}

\subsection{Market Thickness}\label{sec:thickness}

As we already mentioned in the introduction, a feature that makes our loss bound in \Cref{thm:epslowerbound} special is that we operate on a thin market.
The previous section allows us to quantify more precisely what this means in our context.
Combining \Cref{prop:wait-vs-loss} with \Cref{lem:agg_waiting} and \Cref{prop:waitlb}, we obtain

\begin{restatable}{proposition}{PropThickness}\label{prop:thickness}
	Assume that the \dtimename is distributed according to an arbitrary probability measure $\dep$. 
    Then, for large enough $m$, under the greedy algorithm, it holds that
	\begin{align*}
		\frac 1{\Tmax}\int_0^{\Tmax} \E_{\dep}\left[ \psize_s \right] ds \le \frac 65\cdot\frac{m}{d}\text.
	\end{align*}
    Moreover, if $\dep \left( \left[\frac{1}{2d}, \infty\right] \right) > \frac{9}{10}$ and $\Tmax\ge \frac 3d$, then any algorithm $\alg$ satisfies
	\begin{align*}
		\frac 1{\Tmax}\int_0^{\Tmax} \E_{\dep}\left[ \psize_s \right] ds \ge \frac 1{20}\cdot\frac{m}{d}\text.
	\end{align*}
\end{restatable}

Hence, the average pool size of the greedy algorithm is of order $\mathcal O(\frac md)$, i.e., a $\frac 1d$-fraction of the arrival rate.
Moreover, no algorithm can maintain a thinner market for most sojourn times.

For comparison, the average market size under the patient algorithm is of order $\Theta(m)$. This is derived by \citet{ALG20a} under exponentially distributed \dtimenamepl but also holds for any guaranteed sojourn time (because at least half of the agents have to stay until their maximum sojourn times).
Hence, a number of agents equivalent to a constant fraction of the arrival rate are kept in the market in average, ensuring a significantly thicker market.
The greedy algorithm maintains a pool size that is smaller by a magnitude of $d$ compared to the arrival rate.

\section{Analysis of the Patient Algorithm}

In this section, we analyze the loss and waiting time of the patient algorithm.

\subsection{Loss Guarantee of the Patient Algorithm}

The fundamental insight by \citet{ALG20a} is that the patient algorithm outperforms the greedy algorithm in terms of loss if the \dtimenamepl are exponentially distributed.
Even though the greedy algorithm performs very well under other \distnamepl, the patient algorithm still is an alternative worth considering if access to the departure information is not costly.
Here, we show that the patient algorithm still yields exponentially small loss for the case of unit \dtimenamepl.
The central idea of the proof is to use the intuition that the probability of being matched is high if the pool size is large, or, in other words, the market is thick. 
Therefore, the key step is to show that the pool size is moderately large with high probability. 
To this end, we perform a case analysis of the near past of a fixed point in time, considering four time intervals of length $1/3$.
The critical insight is that few agents arriving during the first time interval can be present at the beginning of the last time interval.
To show this, we perform a comparison with an urn model.
 
\begin{restatable}{theorem}{patientlb}\label{lem:patient_lb}
	Assume that the \dtimename is distributed according to the degenerate probability measure $\dep$ with $\dep(\{1\}) = 1$. Then, it holds that
	\begin{align*}
	\loss[\pat] \leq e^{-\frac d5}\text.
	\end{align*}
\end{restatable}

\begin{proof}
	As $\dep$ is fixed throughout the proof, we omit it as subscript. By \Cref{lem:aux-loss} and unit \dtimenamepl, we have
	\begin{align*}
	\loss[\pat] = \limsup_{m,T\to \infty}\frac 1T \int_1^T \p \left( \text{an agent arriving at time $t-1$ perishes} \right) dt\text.
	\end{align*}
	
	Assume that the pool has a size of $k$ at time $t$ and an agent gets critical. Then, the probability that the agent perishes is given by
	$\left( 1- \frac{d}{m} \right)^{k-1}$.
	Thus, in order to prove small loss, it suffices to show that the pool size is large with relatively high probability. Let $t\geq \frac{4}{3}$ be arbitrary. Our goal is to show that $\psize_t>\frac{m}{5}$ with high probability. For this, we condition on the pool size at time~$t-\frac{1}{3}$. By the law of total probability, we have
	\begin{align}
	\p \left( \psize_t \leq \frac{m}{5} \right) = 
	&
	\p \left( \psize_t \leq \frac{m}{5} \ \big| \ \psize_{t-\frac 13} < \frac{m}{8} \right) \p \left( \psize_{t-\frac 13} < \frac{m}{8} \right)\nonumber\\
	&
	+ \p \left( \psize_t \leq \frac{m}{5} \ \big| \ \psize_{t-\frac 13} \geq \frac{m}{8} \right) \p \left( \psize_{t-\frac 13} \geq \frac{m}{8} \right)
	\label{eq:law_of_total_prob}
	\end{align}
	and thus it suffices to bound each of the two summands above. 
	
	We start with the  first one. Each agent that is in the pool at time $t-\frac{1}{3}$ can take at most one additional agent out of the pool when she gets critical. Thus, conditioned on the event where $\psize_{t-\frac 13} < \frac{m}{8}$, it holds that
	$\psize_t \geq \numbarr_{\left( t-\frac{1}{3} , t \right)} - \psize_{t-\frac 13} \geq \numbarr_{\left( t-\frac{1}{3} , t \right)} - \frac{m}{8}$.
	By independence of $\numbarr_{\left( t-\frac{1}{3} , t \right)} $ and $ \psize_{t-\frac 13} $, this already implies that
	\begin{align*}
	\p \left( \psize_t \leq \frac{m}{5} \ \big| \ \psize_{t-\frac 13} < \frac{m}{8} \right)  \leq \p\left(  \numbarr_{\left( t-\frac{1}{3} , t \right)} \leq \frac m5 + \frac m8 \ \big| \ \psize_{t-\frac 13} < \frac{m}{8} \right) & =
	\p\left(  \numbarr_{\left( t-\frac{1}{3} , t \right)} \leq \frac {13}{40} m \right)\\
	& \overset{\text{\Cref{lem:Chernoff}}}{\leq} m^{-2}
	\end{align*}
	for $m$ large enough. This directly implies that
	\begin{align*}
	\p \left( \psize_t \leq \frac{m}{5} \ \big| \ \psize_{t-\frac 13} < \frac{m}{8} \right) \p \left( \psize_{t-\frac 13} < \frac{m}{8} \right) \leq m^{-2}
	\end{align*}
	for $m$ large enough. Now let us consider the bound for the second addend in \eqref{eq:law_of_total_prob}. Here we have that 
	\begin{align}\label{eq:bound total prob2}
	\p \left( \psize_t \leq \frac{m}{5} \ \big| \ \psize_{t-\frac 13} \geq \frac{m}{8} \right) \p \left( \psize_{t-\frac 13} \geq \frac{m}{8} \right) =
	\p \left( \psize_t \leq \frac{m}{5},  \psize_{t-\frac 13} \geq \frac{m}{8} \right)\text.
	\end{align}
	
	We now condition on the pool size at time $t-{\frac{1}{3}}$ and on the number of arriving agents in certain intervals. This is necessary to avoid possible coherences between the evolution of the pool size and the arrival of new agents. For this, define the vector $a=(a_1,a_2,a_3)\coloneqq\left( \numbarr_{\left[t-\frac{4}{3},t-1\right)} , \numbarr_{\left[t-1,t-\frac{2}{3}\right)}, \numbarr_{\left[t-\frac{2}{3},t-\frac{1}{3}\right)}  \right)$ to get
	\begin{align}\label{eq:conditioning on arrivals}
	&\p \left( \psize_t \leq \frac{m}{5},  \psize_{t-\frac 13} \geq \frac{m}{8} \right) 
	= 
	\sum_{k_1,k_2,k_3 =0}^{\infty} \sum_{l \geq \frac{m}{8}}
	\p \left( \psize_t \leq \frac{m}{5},  \psize_{t-\frac 13} = l , a=(k_1,k_2,k_3) \right)\\
	& \notag
	=
	\sum_{k_1,k_2,k_3 = \lfloor m/3-m^{2/3} \rfloor}^{\lfloor m/3+m^{2/3} \rfloor} \sum_{l \geq \frac{m}{8}}
	\p \left( \psize_t \leq \frac{m}{5},  \psize_{t-\frac 13} = l ,  a=(k_1,k_2,k_3) \right) \\
	& \notag + \sum_{\substack{(k_1,k_2,k_3) \in \N^3 : \\ (k_1,k_2,k_3) \notin \left[ \lfloor m/3-m^{2/3} \rfloor , \lfloor m/3+m^{2/3} \rfloor \right]^3}} \sum_{l \geq \frac{m}{8}}
	\p \left( \psize_t \leq \frac{m}{5},  \psize_{t-\frac 13} = l ,  a=(k_1,k_2,k_3) \right) \ .
	\end{align}
	The second summand in the above sum can be bounded by
	\begin{align*}
		& \sum_{\substack{(k_1,k_2,k_3) \in \N^3 : \\ (k_1,k_2,k_3) \notin \left[ \lfloor m/3-m^{2/3} \rfloor , \lfloor m/3+m^{2/3} \rfloor \right]^3}} \sum_{l \geq \frac{m}{8}}
		\p \left( \psize_t \leq \frac{m}{5},  \psize_{t-\frac 13} = l ,  a=(k_1,k_2,k_3) \right) \\
		& \leq \sum_{\substack{(k_1,k_2,k_3) \in \N^3 : \\ (k_1,k_2,k_3) \notin \left[ \lfloor m/3-m^{2/3} \rfloor , \lfloor m/3+m^{2/3} \rfloor \right]^3}}
		\p \left(  a=(k_1,k_2,k_3) \right)\\
		& \leq \sum_{\substack{k \in \N : \\ k \notin \left[  m/3-m^{2/3} , m/3+m^{2/3} \right]}}
		\p \left(  a_1 = k \right) + 
		\p \left(  a_2 = k \right) + 
		\p \left(  a_3 = k \right) \\
		& =	3 \sum_{\substack{k \in \N : \\ k \notin \left[  m/3-m^{2/3} , m/3+m^{2/3} \right]}}
		\p \left(  a_1 = k \right) 
		\overset{\text{\Cref{lem:Chernoff}}}{\leq} 6e^{-m^{1/3}} \leq m^{-2}\text.
	\end{align*}
	The equality follows because $a_1$, $a_2$, $a_3$ are identically distributed, and the last two inequalities hold for large enough $m$. 
	Inserting this into \eqref{eq:conditioning on arrivals}, we see that
	\begin{align*}
		\notag &\p \left( \psize_t \leq \frac{m}{5},  \psize_{t-\frac 13} \geq \frac{m}{8} \right) \leq 
		\sum_{\substack{k_1,k_2,k_3\\=\lfloor m/3-m^{2/3} \rfloor}}^{\lfloor m/3+m^{2/3} \rfloor} \sum_{l \geq \frac{m}{8}}
		\p \left( \psize_t \leq \frac{m}{5},  \psize_{t-\frac 13} = l ,  a=(k_1,k_2,k_3) \right) + m^{-2}
	\end{align*}
	for all large enough $m$. Our next goal is to show that
	\begin{align}\label{eq:toshow patient1}
		& \p \left( \psize_t \leq \frac{m}{5},  \psize_{t-\frac 13} = l ,  a=(k_1,k_2,k_3) \right) 
		\leq m^{-2} \p \left( \psize_{t-\frac 13} = l ,  a=(k_1,k_2,k_3) \right) \text,
	\end{align}
	as this already implies that
	\begin{align*}
		\notag &\p \left( \psize_t \leq \frac{m}{5},  \psize_{t-\frac 13} \geq \frac{m}{8} \right) \leq 
		\sum_{\substack{k_1,k_2,k_3\\=\lfloor m/3-m^{2/3} \rfloor}}^{\lfloor m/3+m^{2/3} \rfloor} \sum_{l \geq \frac{m}{8}}
		\p \left( \psize_t \leq \frac{m}{5},  \psize_{t-\frac 13} = l ,  a=(k_1,k_2,k_3) \right) + m^{-2}\\
		& \leq  \sum_{\substack{k_1,k_2,k_3\\=\lfloor m/3-m^{2/3} \rfloor}}^{\lfloor m/3+m^{2/3} \rfloor} \sum_{l \geq \frac{m}{8}} m^{-2}
		\p \left( \psize_{t-\frac 13} = l ,  a=(k_1,k_2,k_3) \right) + m^{-2}
		\leq 2m^{-2}\text.
	\end{align*}
	By the law of conditional probability, it suffices to show
	\begin{align}\label{eq:toshow patient2}
		\p \left( \psize_t \leq \frac{m}{5} \ \big| \ \psize_{t-\frac 13} = l ,  a=(k_1,k_2,k_3) \right)
		\leq m^{-2} 
	\end{align}
	for all $k_1, k_2, k_3 \in \left[ \lfloor m/3-m^{2/3} \rfloor ,  \lfloor m/3 + m^{2/3} \rfloor \right]$ and $l \geq \frac{m}{8}$, because this already implies \eqref{eq:toshow patient1}.
	
	For $i\in \{1,2,3\}$, let $v_{1}^i, v_2^i, \ldots, v_{k_i}^i$ be the $k_i$ many agents that arrive during the time interval  $\left[t-\frac{5-i}{3},t-\frac{4-i}{3}\right)$. We define the random variable $K_1$ by
	 \begin{align*}
	K_1 = \sum_{j=1}^{k_1} \mathbf{1}_{\left\{ v_{j}^1 \text{ is in the pool at time } t-\frac{1}{3} \right\}}\text,
	 \end{align*}
	where~$\mathbf 1$ denotes the indicator function. From our conditions, we already know that $K_1 \leq k_1$. 
	
	In the next step, we make use of the following observation: For each agent $x$ that arrives during the interval $\left[t-\frac{4}{3},t-1\right)$ and each agent $y$ that arrives during the time interval $\left[t-1,t-\frac{1}{3}\right)$, the chances of agent $y$  of making it to the pool at time $t-\frac{1}{3}$ are at least as high as for agent $x$, as the agent $x$ witnesses at least as many critical agents as $y$. 
	Moreover, since $k_1,k_2,k_3 \in \left[\lfloor \frac{m}{3} -m^{2/3} \rfloor, \lfloor \frac{m}{3} +m^{2/3} \rfloor\right]$, it holds that $\frac{k_1}{k_1+k_2+k_3} \leq \frac 25$ for $m$ large enough. 
	We will assume that $\frac{k_1}{k_1+k_2+k_3} \leq \frac 25$ henceforth.
	Thus, we can infer that
	\begin{align*}
	\E\left[ K_1 \ \big| \ \psize_{t-\frac{1}{3}} = l ,  a=(k_1,k_2,k_3) \right] \leq l \frac{k_1}{k_1+k_2+k_3} \leq \frac 25 l\text.
	\end{align*}
	
	We want to translate this bound on $K_1$, that holds in expectation, to a bound that holds with high probability, i.e., our next goal is to show that $K_1 < l/2$ with high probability. The distribution of $K_1$ highly depends on the arrival/departure-process, as agents that arrive at a later point in time typically witness less critical agents up to time $t-\frac{1}{3}$. 
	We expect $K_1$ to be maximal, when all agents getting critical in the time frame $\left[t-\frac{4}{3} , t-\frac{1}{3} \right)$ get critical just before time step $t-\frac{1}{3}$, as all agents arriving in the time frame $\left[t-\frac{4}{3} , t-\frac{1}{3} \right)$ have the same chance of making it to the pool $\pool_{t-\frac{1}{3}}$ in this case. 
	In all other cases, agents that arrive earlier always have a lower chance of making it to the pool. 
	Thus, we can bound the probability of $K_1>\frac{l}{2}$ with the outcome of a different experiment, namely an urn experiment, that we describe now.
	
	Let $n=k_2+k_3$ and consider an urn with $N=k_1+k_2+k_3=k_1+n$ balls in it, where $k_1$ of them are red and $k_2+k_3$ of them are blue. We draw $l\leq N$ balls out of this urn without replacement. Let $\Tilde{K_1}$ be the number of red balls taken. 
	We defer the technical proof of the next steps to two lemmas shown in \Cref{app:patient}.

\begin{restatable}{lemma}{stochdom}\label{lem:stochdom}
	$\tilde{K_1}$ stochastically dominates $K_1$.
\end{restatable}

    This lemma implies that we only have to bound the probability that $\tilde{K_1} \geq l/2$.
    The crucial step for this is done in our second lemma.

\begin{restatable}{lemma}{urnbound}\label{lem:urnbound}
    It holds that $\p\left( \tilde{K_1} \geq \frac{l}{2} \right) \leq 2m \, 0.98^{m/8}$.
\end{restatable}

	Hence, combining \Cref{lem:stochdom,lem:urnbound}, it holds that 
	 \begin{align*}
		\p\left( K_1 \geq \frac{l}{2} \right) \leq \p\left( \tilde{K_1} \geq \frac{l}{2} \right) \leq 2m 0.98^{m/8}
		\leq 0.998^m \leq m^{-3},
	 \end{align*}
	where the last two inequalities hold for $m$ large enough. If $K_1 \leq \frac{l}{2}$, then	
		$\psize_t \geq \psize_{t-\frac 13} - 2 K_1 + \numbarr_{(t-\frac{1}{3},t)} \geq \numbarr_{(t-\frac{1}{3},t)}$.
	So, we can conclude in particular that
	\begin{align*}
		&\p \left( \psize_t \leq \frac{m}{5} \ \big| \ \psize_{t-\frac 13} = l ,  a=(k_1,k_2,k_3) \right) \\
		&
		\leq
		\p \left( K_1 > \frac{l}{2} \ \big| \ \psize_{t-\frac 13} = l ,  a=(k_1,k_2,k_3) \right)
		+
		\p \left( \numbarr_{(t-\frac{1}{3},t)} < \frac{m}{5}  \ \big| \ \psize_{t-\frac 13} = l ,  a=(k_1,k_2,k_3) \right) \\
		& \overset{\text{\Cref{lem:Chernoff}}}{\leq} m^{-3} + 	\p \left( \numbarr_{(t-\frac{1}{3},t)} < \frac{m}{5}  \right) \leq m^{-2}
	\end{align*}
	for $m$ large enough, which finally shows \eqref{eq:toshow patient2}. For the second inequality we used that $\numbarr_{(t-\frac{1}{3},t)}$ is independent of $\psize_{t-\frac 13}$ and $a=(a_1,a_2,a_3)$. Together with \eqref{eq:law_of_total_prob}, this already implies that
	\begin{align*}
	\p\left( \psize_t \leq \frac{m}{5} \right) \leq m^{-1}
	\end{align*}
	for $m$ large enough. With this result we get that
	\begin{align*}
	&\p \left( \text{an agent arriving at time $t-1$ perishes } \right)\\
	&\leq \p \left( \text{an agent getting critical at time $t$ perishes } \right)\\
	&= \p \left( \text{an agent getting critical at time $t$ perishes } \ \big| \ \psize_t\geq \frac{m}{5} \right) \p\left( \psize_t\geq \frac{m}{5} \right) \\
	&
	+ 
	\p \left( \text{an agent getting critical at time $t$ perishes } \ \big| \ \psize_t < \frac{m}{5} \right) \p\left( \psize_t < \frac{m}{5} \right)\\
	&
	\leq \left(1-\frac{d}{m}\right)^{\frac m5 -1} + m^{-1}
	\underset{m \to \infty}{\longrightarrow} e^{-\frac d5}.
	\end{align*}
This proves the desired loss bound of $\pat$.
\end{proof}

We remark that \Cref{lem:patient_lb} also holds for an arbitrary constant sojourn time. 
In fact, \Cref{lem:timechange}, which deals with time shift, also holds for the patient algorithm (with the identical proof).
Hence, we obtain the following corollary. 

\begin{corollary}
    Let $k>0$. 
    Assume that the \dtimename is distributed according to the degenerate probability measure $\dep$ with $\dep(\{k\}) = 1$. 
    Then, it holds that
	\begin{align*}
	\loss[\pat] \leq e^{-\frac {kd}5}\text.
	\end{align*}
\end{corollary}

\subsection{Waiting Time of the Patient Algorithm}

For the patient algorithm, at least half of the agents are held in the market for their maximum possible sojourn.
Indeed, we can split the agents into those that stay until they would perish and those that get matched by those agents.
We can then pair each agent in the latter set with the unique agent in the former set that they are matched with.
Hence, under unit \dtimenamepl and for arrivals until time $\Tmax-1$, the expected waiting time is at least $1/2$, a fixed constant.
In our next result, we generalize this observation for a general distribution of the maximum sojourn time.
We obtain a constant average waiting time, where the constant depends on how much mass the \distname has close to a \dtimename of~$0$.
The idea of the proof of this statement follows a similar reasoning as the above pairing.
We divide the set of agents into a large set of agents whose maximum sojourn time is high and a small set of agents with low maximum sojourn time.
At least half the agents in the first set which do not match with an agent in the second set have to sojourn in the market for a long time.
	
\begin{restatable}{proposition}{wtPatLb}\label{prop:pattime}
	Let $h > 0$ and assume that the \dtimename is distributed according to a probability measure $\dep$ with $\dep \left( \left[h, \infty\right] \right) > \frac{4}{5}$. Then, for sufficiently large $m$ and $\Tmax$, it holds that
	\begin{align}\label{eq:patient waiting lowerbound}
	\waitobj[\pat](m,\Tmax) \geq \frac{h}{4} \text.
	\end{align}
\end{restatable}

\begin{proof}
	Assume that the agents $v_1,\ldots,v_K$ arrive up to time $T-h$, and the agents $u_1,\ldots,u_L$ arrive in the time interval $\left(T-h,T\right]$. 
	We will assume that $20L\leq K$, which holds with high probability for $T$ large enough.
	
	We will first focus on the agents $v_1,\ldots,v_K$. 
	The expected number among them with a maximal sojourn time in $ \left[h, \infty\right]$ is $\dep \left( \left[h, \infty\right] \right)K$. We call such an agent an agent with high maximal sojourn. By the strong law of large numbers there will be at least $\frac{4K}{5}$ such agents with high probability, as $K \to \infty$. We denote by $\mathcal{A}(m,T)$ the event that $20L\leq K$ and that at least $\frac{4K}{5}$ agents have a high maximal sojourn. Then $\limsup_{m,T\to \infty} \p \left(\mathcal{A}(m,T)\right) = 1$.	
	 Assume that there are at least $\frac{4K}{5}$ agents with high maximal sojourn. At most $\frac{K}{5}+L$ of them can leave the pool by matching with an agent that does not have a high maximal sojourn or arrived after time $T-h$.
	So at least $\frac{4K}{5} - \left( \frac{K}{5} + L \right) =\frac{3K}{5}-L$ agents of $v_1,\ldots,v_K$ that have a high maximal sojourn need 
    to stay in the market for a time of at least $h$, or they match with another agent that of $v_1,\ldots,v_K$ with high maximal sojourn. Since the latter case can only affect half of them, we obtain a total waiting time of at least 
	\begin{align*}
		\frac{1}{2}\left(\frac{3K}{5}-L\right)h \geq 0.275Kh  
	\end{align*}
	where we used the assumption $20L\leq K$ in the last step. So in particular we have that
	\begin{align}\label{eq:intermediate claim}
		\waitobj[\pat](m,T)  
		\geq
		\frac 1{m\Tmax}\E \left[0.275 Kh  \mathbf{1}_{\mathcal{A}(m,T)} \right]
	\end{align}

    The division by $m\Tmax$ is for getting from total to average waiting time.
    
    Now, $K$ is a random variable with Poisson$(m(T-h))$ distribution. So in particular its second moment $\E\left[K^2\right]$ is given by $(m(T-h))^2+m(T-h)$. Thus we get by Cauchy-Schwarz that
	\begin{align*}
		&\E \left[0.275 Kh  \mathbf{1}_{\mathcal{A}(m,T)^C} \right]
		\leq h \sqrt{\E\left[K^2\right]}\sqrt{\p \left(\mathcal{A}(m,T)^C\right)}\\
		&\leq h \sqrt{ (m(T-h))^2+m(T-h) } \sqrt{\p \left(\mathcal{A}(m,T)^C\right)}
        \\
        &
        \leq h \sqrt{ 2 (m(T-h))^2 } \sqrt{\p \left(\mathcal{A}(m,T)^C\right)}
        \leq 2h m(T-h)  \sqrt{\p \left(\mathcal{A}(m,T)^C\right)},
	\end{align*}
    where the second to last inequality holds for all $T$ large enough. Inserting this into \eqref{eq:intermediate claim}, we see that
	\begin{align*}
		&\waitobj[\pat](m,T)  \geq \frac 1{m\Tmax}\E \left[0.275 Kh  \mathbf{1}_{\mathcal{A}(m,T)} \right]\\
		&
		=
		\frac 1{m\Tmax}\left(\E \left[0.275 Kh  \right] - \E \left[0.275 Kh  \mathbf{1}_{\mathcal{A}(m,T)^C} \right]\right) \\
		&
		\geq \frac 1{m\Tmax}\left(0.275 h \E \left[K \right]  - 2 h m(T-h) \sqrt{\p \left(\mathcal{A}(m,T)^C\right)}\right)
        \\
        &
        =
        \frac {h}{m\Tmax}\left(0.275 m (T-h)  - 2 h m(T-h) \sqrt{\p \left(\mathcal{A}(m,T)^C\right)}\right)
        \geq
		0.25 h = \frac{h}{4},
	\end{align*}

    We use $\E \left[K \right] = m (T-h)$.
    Moreover, the last inequality holds for $m,T$ large enough, since $\p \left(\mathcal{A}(m,T)^C\right) \to 0$ for $m,T \to \infty$.
\end{proof}

\section{Simulations}

In the previous sections, we have provided theoretical bounds on the loss of the greedy and patient algorithm, particularly for the case of unit \dtimenamepl.
These bounds are guaranteed to hold for $m$ and $T$ tending to infinity and give an indication on the market behavior for fixed $m$ or $T$.
In order to complete the picture, we provide simulations for rather small matching markets. 
Note that such markets are often of large scale in reality, see, e.g., \citet{Unve10a} for a discussion of kidney exchange markets and \citet{BCFY14a} for a discussion of child adoption markets. 
Their respective data on the pool size (greater than $50\,000$) is confirmed by recent numbers from the Organ Procurement and Transplantation Network and the US Census 2020, respectively.
To be conservative, we run our simulations for arrival rates of $m = 500$ and $m = 1000$.

First, we investigate our two key objectives, i.e., 
the percentage of unmatched agents and the average waiting time of agents.
In particular, this analysis gives insights on the significance of error terms in practical examples.
Moreover, we compare the loss of the greedy algorithm with a variant of the greedy algorithm where the choice of the matching partner is not uniformly at random.

\begin{figure}
    \centering
	\begin{subfigure}[t]{0.47\textwidth}
	\begin{tikzpicture}
		\begin{axis}[
			xlabel = \dparaname,
			ymode = log,
			ylabel = loss,
			xmin = 0,
			xmax = 15,
			ymin = 1e-6,
			ymax = 1,
			width = \textwidth,
			legend style={font=\tiny},
		]
		\addplot[
			domain = 0:15,
			samples = 21,
			color = red, 
			style = dashed]
			{exp(-x/(2*ln(2)))};
		\addlegendentry{GDY ub}
		\addplot[
			color=red,
			mark=o, 
			only marks
			]
			coordinates {
				(1, 0.2856215783816214)(2, 0.12320273815034577)(3, 0.057778503220431314)(4, 0.027594554311836668)(5, 0.013360191811966633)(6, 0.006591974308818731)(7, 0.003141411524909763)(8, 0.001560165443586359)(9, 0.0007545623538347758)(10, 0.0003769513732728478)(11, 0.00017497042999733045)(12, 8.799102491545862e-05)(13, 4.203875973647703e-05)(14, 2.4989679262464602e-05)(15, 1.0994854408136993e-05)
			};
		\addlegendentry{GDY}
		\addplot[
			color=blue,
			mark=x, 
			only marks
			]
			coordinates {
				(1, 0.2860570338632315)(2, 0.12215991443508162)(3, 0.05790535550231477)(4, 0.02750520398419427)(5, 0.013119507311595042)(6, 0.0065560008394879125)(7, 0.0032426804423020127)(8, 0.001600657459303709)(9, 0.0006896710269201591)(10, 0.00036888970197910826)(11, 0.00020318611848453184)(12, 9.497207820900656e-05)(13, 4.59726003302032e-05)(14, 2.599285196570943e-05)(15, 1.0021074319293474e-05)
			};
		\addlegendentry{PAT}
		\addplot[
			domain = 0:15,
			samples = 21,
			color = red,
			style = dotted]
			{.5*exp(-x)};
		\addlegendentry{GDY lb}
		\addplot[
			domain = 0:15,
			samples = 21,
			color = blue,
			style = dotted]
			{exp(-2*x)};
		\addlegendentry{lb}
			
		\end{axis}
	\end{tikzpicture}
	\caption{
		Logarithmic plot of average proportion of unmatched agents for $\gdy$ and $\pat$.
	\label{fig:greedy_patient_loss}}
	\end{subfigure}
	\hfill
	\begin{subfigure}[t]{0.47\textwidth}
	\begin{tikzpicture}
		\begin{axis}[
			boxplot/draw direction=y,
			xtick = {3,4,5,6,7},
			ymin = 0,
			xlabel = \dparaname,
			ylabel = loss,
			width = \textwidth
		]
			\addplot+ [boxplot prepared={
				draw position=3,
				lower whisker = 0.05338577696960114,
				lower quartile = 0.05601562972215509,
				median = 0.05694610667113194,
				upper quartile = 0.057595077113308774,
				upper whisker = 0.0595920008039393},
				color = red
			] coordinates {};
			\addplot+ [boxplot prepared={
				draw position=4,
				lower whisker = 0.024985508405125027,
				lower quartile = 0.026899763282179235,
				median = 0.02733426515468061,
				upper quartile = 0.02800195504062277,
				upper whisker = 0.030217129071170086},
				color = red
			] coordinates {};
			\addplot+ [boxplot prepared={
				draw position=5,
				lower whisker = 0.011682056997293425,
				lower quartile = 0.012929155409674842,
				median = 0.013230723044868102,
				upper quartile = 0.013630579136600558,
				upper whisker = 0.01495473432763916},
				color = red
			] coordinates {};
			\addplot+ [boxplot prepared={
				draw position=6,
				lower whisker = 0.005507834247991728,
				lower quartile = 0.0061464752525925564,
				median = 0.006418996781636926,
				upper quartile = 0.006669912376939353,
				upper whisker = 0.0073681283762313315},
				color = red
			] coordinates {};
			\addplot+ [boxplot prepared={
				draw position=7,
				lower whisker = 0.00230078226597043,
				lower quartile = 0.002962994618259492,
				median = 0.0031252568037668767,
				upper quartile = 0.0033058182619123276,
				upper whisker = 0.003891756734957251},
				color = red
			] coordinates {};
		\end{axis}
	\end{tikzpicture}
	\caption{
		Box-and-whisker plots describing average proportion of unmatched agents for $\gdy$ under small \dparaname.
		\label{fig:greedy_loss_boxplot}}
\end{subfigure}
\begin{subfigure}[t]{0.47\textwidth}
\begin{tikzpicture}
	\begin{axis}[
		xlabel = \dparaname,
		ymode = log,
		ylabel = loss,
		xmin = 0,
		xmax = 12,
		ymin = 1e-6,
		ymax = 1,
		width = \textwidth,
		legend style={font=\tiny},
	]
	\addplot[
		color=red,
		mark=o,
		only marks
		]
		coordinates {
			(1, 0.28696756794122097)(2, 0.12327469864353378)(3, 0.05775488720600714)(4, 0.027561581722109044)(5, 0.013493389728193299)(6, 0.006514332530698293)(7, 0.0031322647253421605)(8, 0.0014410313995257063)(9, 0.0008060588422954876)(10, 0.00034288239133977046)(11, 0.00019520770099385747)(12, 9.09881715377001e-05)(13, 3.3998878037024776e-05)(14, 2.00160128102482e-05)(15, 4.999580035277037e-06)
		};
	\addlegendentry{GDY random}
	\addplot[
		color=green,
		mark=square,
		only marks
		]
		coordinates {
			(1, 0.2776922246177107)(2, 0.10904771431429429)(3, 0.0440351913662995)(4, 0.017269250877885303)(5, 0.006517210633848713)(6, 0.0024149901459598437)(7, 0.0009439467786189548)(8, 0.0003356633296803306)(9, 0.00012483796032749492)(10, 3.804885472947264e-05)(11, 1.7009933801339983e-05)(12, 5.002150924897706e-06)(13, 9.99841025276981e-07)
		};
	\addlegendentry{GDY sojourn}
	\addplot[
		domain = 0:15,
		samples = 16,
		color = red,
		style = dotted]
		{.5*exp(-x)};
	\addlegendentry{GDY lb}
	\addplot[
		domain = 0:15,
		samples = 16,
		color = blue,
		style = dotted]
		{exp(-2*x)};
	\addlegendentry{lb}
		
	\end{axis}
\end{tikzpicture}
\caption{Logarithmic plot of average proportion of unmatched agents for $\gdy$ under different match selection. 
\label{fig:random_first_exit_loss}}
\end{subfigure}
\hfill
\begin{subfigure}[t]{0.47\textwidth}
	\centering
	\begin{tikzpicture}
		\begin{axis}[
			xlabel = \dparaname,
			ylabel = waiting time,
			xmin = 0,
			xmax = 20,
			ymin = 1,
			ymax = 30,
			ytick = {5,10,15,20,25,30},
			yticklabels = {$5^{-1}$,$10^{-1}$,$15^{-1}$,$20^{-1}$,$25^{-1}$,$30^{-1}$},
			y dir = reverse,
			width = \textwidth,
			legend style = {font=\tiny},
		]
		\addplot[
			domain = 0:20,
			samples = 21,
			color = red, 
			style = dashed]
			{x*5/6};
		\addlegendentry{GDY ub}
		\addplot[
			color=red,
			mark=o, 
			only marks
			]
			coordinates {
				(1, 2.267444246134425)(2, 3.455865756976548)(3, 4.7185242864597985)(4, 6.023168672164145)(5, 7.378060027240382)(6, 8.731340369840217)(7, 10.157193395844041)(8, 11.5781559301739)(9, 12.979319850436383)(10, 14.45923199975333)(11, 15.88588911602022)(12, 17.33831739886411)(13, 18.804432738739912)(14, 20.225333234665182)(15, 21.703782250844842)(16, 23.164956154043107)(17, 24.55578191543234)(18, 26.052769806908653)(19, 27.48920594046567)(20, 29.03035625294548)
			};
			\addlegendentry{GDY}
			\addplot[
				domain = 0:20,
				samples = 21,
				color = red, 
				style = dotted]
				{x*20};
			\addlegendentry{lb}
		\end{axis}
	\end{tikzpicture}
	\caption{
		Inverse linear plot of average waiting time for $\gdy$.
		\label{fig:waiting_time}}
\end{subfigure}

\caption{Visualization of results of simulations. 
Each datapoint in the first, third, and fourth figures corresponds to $10$ runs for $m=1000$ and $T=100$. Each box-and-whisker plot in the second figure represents $200$ runs for $m=500$ and $T=100$. 
	Note that the linear behavior in the logarithmic plots of the first and third figure means that the decay is exponential.
	For reference, we add our theoretically obtained bounds in the large market limit, i.e., the upper bound on the performance of $\gdy$ (cf.~\Cref{thm:greedy_loss_upper_bound}), and our general (cf.~\Cref{thm:lossgenlb}) and $\gdy$-specific (cf.~\Cref{cor:greedylbconst}) lower bounds obtained in \Cref{sec:gdyanalysis}.\label{fig:simu}}
\end{figure}

We run an event-driven simulation with two event types (arrivals, critical times). At each step we process the next event time; almost surely only one event occurs at a time.
We start with an empty pool at time $0$, i.e., no agents are present initially.
Then, agents join the pool at Poisson rate $m$.
This means in particular that the difference between the arrival time of two successively arriving agents is distributed with respect to an exponential distribution with expectation $\frac{1}{m}$, i.e., the arrival of the next agent can always be simulated by evaluating a random variable.
Every agent stays for a unit \dtimename.
During the simulation, a matching algorithm is employed and used to determine pairs of agents from the pool to be matched.
Our focus lies on the greedy algorithm.
We evaluate statistics up to time $T$; with $T = 100$ and unit \dtimenamepl, boundary effects are negligible.

The results of our simulations are depicted in \Cref{fig:simu}. 
First, we consider the average proportion of unmatched agents for $\gdy$ and $\pat$ in \Cref{fig:greedy_patient_loss}. 
The simulations indicate an exponentially small loss with respect to~$d$ for both algorithms, which is in line with the bounds obtained in our theoretical findings. 
Interestingly, the average proportion of unmatched agents in both algorithms appears to be surprisingly similar, hinting at a possible connection between two basic algorithms using neither structural information of the compatibility relation nor timing information beyond arrivals and critical states of agents.
We provide an intuition for this observation in \Cref{sec:gdyEQUpat}.

In \Cref{fig:greedy_loss_boxplot}, we complement this aggregated view by displaying 
the distribution of unmatched agents of the greedy algorithm under small \dparanamepl by means of box-and-whisker plots.\footnote{
The whiskers denote the minimum and maximum of the 
proportion of unmatched agents, respectively, while the lower and upper bounds of the box represent the first and third quartile.
Finally, the middle line corresponds to the median.}
We can see that there is little spread, 
indicating the robustness of the theoretically proved statements for very sparse markets.

Next, we consider the assumption of $\gdy$ to perform matching decisions uniformly at random among the set of compatible agents.
In principle, it is also possible to use further structural or timing information in order to perform a selection.
If the social planner receives additional information on the remaining \dtimename of individual agents, one natural matching choice is to use the \emph{sojourn-based tie-breaking} rule, which selects the compatible agent with the least remaining \dtimename as the matching partner.\footnote{This agent is almost surely unique.}
This may avoid the perishing of an agent with small remaining sojourn time.
In fact, this algorithm performs even closer to our lower bound for the necessary loss of greedy algorithms of $\frac{1}{2} e^{-d}$, as can be seen in \Cref{fig:random_first_exit_loss}.
Note that this lower bound does not depend on the uniformly random selection of matched pairs, but even holds for an arbitrary tie-breaking mechanism when matching under $\gdy$.

Finally, we investigate our second objective, the average waiting time of agents for the greedy algorithm.
Our results are depicted in \Cref{fig:waiting_time}. 
Note that the upper bound on the waiting time in \Cref{lem:agg_waiting} already holds for relatively small arrival rates, approximately at $m\approx 125$. 
Our simulations show that the aggregated waiting time is consistently below the expected upper bound, again demonstrating the superiority of the greedy algorithm over the patient algorithm with respect to waiting time minimization.

\section{Discussion and Conclusion}\label{sec:conclusion}

We have analyzed instantaneous matching as captured by the greedy algorithm and matching at the moment of perishing as captured by the patient algorithm. 
We found that, beyond market thickness, the performance of the greedy algorithm strongly depends on the distribution of the \dtimenamepl. 
In \Cref{table:comparison}, we display loss bounds of the greedy and patient algorithm.
We showcase our results for the important special case of unit waiting times, and compare our bounds with the bounds for exponentially distributed waiting times considered by \citet{ALG20a}.
In particular, if agents are guaranteed to stay in the market for a minimum amount of time, then the greedy algorithm achieves an exponentially small loss.
This is striking because we operate in a thin market, whereas existing guarantees critically relied on market thickness (see \Cref{sec:thickness} for a discussion of market thickness in our model). 
In line with this, \citet{ALG20a} argue that information gain is essential because it can be used to achieve market thickness.
By contrast, we demonstrate that coarse knowledge about agent behavior can substitute for detailed departure information to achieve strong performance guarantees.

It is interesting to carve out more precisely the reason why the loss of the greedy algorithm is sensitive to the sojourn distribution.
Under $\gdy$, at least half of the agents do not get matched at their arrival (because every matched pair consists of a newly arrived agent and an agent from the pool).
Hence, loss is caused by the agents entering the pool.
For these, all agents present at their arrival are secondary because they are incompatible, and hence their matching opportunities depend on newly arriving agents.
Now, when comparing sojourn distributions like an exponential distribution with mean~$1$ and unit sojourn times, they lead to a similar pool size and an identical mean sojourn.
Moreover, compatibility is identical.
Hence, for every agent arriving during their sojourn, the probability of getting matched is also similar.
Thus, whether an agent entering the pool gets matched depends more crucially on the mere number of agents arriving during their sojourn.
This is essentially captured by the sojourn distribution.
Under a distribution like the exponential distribution, a rather long sojourn is possible, and such agents will match with a very high probability.
However, a very short sojourn also occurs with moderately high probability, and such agents are likely to perish.
By contrast, a guaranteed sojourn time leads to a guarantee on the number of matching opportunities, which avoids a high loss probability from ever occurring.

In addition, we believe that there are significant reasons to promote usage of the greedy algorithm beyond its good performance with respect to both of our measures.
First, the process of gaining knowledge of exact \dtimenamepl might be costly or lead to ethical concerns \citep{RBG15a}.
This is avoided when using an algorithm like the greedy algorithm that does not rely on such information. 

Second, it is debatable whether market thickness is desirable at all.
A large pool not only increases the waiting time but also indicates a backlog of unmatched agents and the risk of operational congestion.
In fact, the analysis of real-life data indicates that market thickness may lead to a reduced compatibility of agents and hence a worse outcome \citep{Fong20}.
Moreover, a thicker market might lead to a measurably higher loss:
In a quasiexperiment on a holiday property rental platform, \citet{LiNe20a} find that a sudden boost in market thickness leads to a loss of 5.6\% matches.
They explain this with search friction, i.e., figuring out good options in a limited time window.

This observation follows a general paradigm that market inefficiencies can be caused by an abundance of matching possibilities \citep{Roth18a}.
Even though our model does not explicitly contain a cost for identifying the compatibility of agents, these studies promote that maintaining a manageable pool size can be important.
Hence, using the greedy algorithm might yield low market congestion as an additional beneficial side effect.

To conclude the paper, we elaborate on two important aspects of our work that open up prospects for further investigation. 
First, we consider the loss equivalence of the greedy and patient algorithm under unit waiting times as observed in our simulations. 
Second, we discuss possible extensions of our work to the case of heterogeneous agents.

\begin{table*}[t!]\centering
		\caption{Bounds for the loss of $\gdy$ and $\pat$ under different \dtimenamepl. The bounds for unit \dtimenamepl are novel, the bounds for exponential \dtimenamepl are due to \citet[Theorem~4]{ALG20a}. The loss of $\gdy$ is exponentially small for unit \dtimenamepl.}\label{table:comparison}	
		\begin{tabular}{@{}lllllll@{}}
			\toprule
			& 
			\multicolumn{1}{c}{} & 
			\multicolumn{2}{c}{unit \dtimename} & 
			\phantom{a}& 
			\multicolumn{2}{c}{exponential \dtimename}
			\\ 
			\cmidrule{3-4} 
			\cmidrule{6-7}
			\alg && lower bound & upper bound && lower bound & upper bound
			\\
			\midrule
			\gdy && $\frac 12 e^{-d}$ & $e^{-\frac d{2\ln(2)}}$ && $\frac 1 {2d+1}$ & $\frac{\ln 2}d$ \\
			\midrule
			\pat && $e^{-2d}$ & $e^{-\frac d5}$ && $\frac {e^{-d}}{d+1}$ & $\frac 12 e^{-\frac d2}$ 
			\\
			\bottomrule
		\end{tabular}

\end{table*}

\subsection{Equivalence of Greedy and Patient Algorithm}\label{sec:gdyEQUpat}

While our theoretical results only show an exponential loss in both cases, \Cref{fig:greedy_patient_loss} suggests that the loss of the two algorithms is (approximately) identical. Indeed, there is some intuition why the loss should be approximately $\frac 12 e^{-\frac{d}{2 \ln(2)}}$ for both the greedy and the patient algorithm. In this section, we provide two different, non-rigorous arguments which lead to exactly the loss observed in our simulations in both cases. Unfortunately, these arguments are heavily based on steady-state analysis and are therefore hard to be made precise for {\sl non}-Markovian processes. Identifying a relationship of the two algorithms that establishes an identical loss is an intriguing problem for further research.

\paragraph*{Greedy algorithm} We start by reasoning about the pool size similar to the proof sketch in \Cref{sect:gdysketch}. Let $z_{\mathit{Eq}}$ denote the pool size at equilibrium. Assuming that the loss is negligible, half of the agents instantaneously match at arrival. The probability of not forming an edge at arrival is about
\begin{align*}
	\left(1-\frac{d}{m}\right)^{z_{\mathit{Eq}}} \approx e^{-d \frac{z_{\mathit{Eq}}}{m}} \overset{!}{=} \frac{1}{2}
\end{align*}
which suggests that the pool in equilibrium should be of size $\frac{\log(2)}{d}m$. If an agent $v$ does form an edge at arrival, but joins the pool, approximately $m$ agents will arrive during her maximum \dtimename and approximately $\frac{m}{2}$ of these agents will match at arrival. The probability that a specific agent, matching at arrival, matches with $v$ equals $\frac{1}{z_{\mathit{Eq}}}$, by symmetry. So assuming that $z_{\mathit{Eq}}$ stays approximately constant over time, the probability that $v$ perishes can be estimated as
\begin{align*}
	\left( 1- \frac{1}{z_{\mathit{Eq}}} \right)^{\frac{m}{2}} \approx 
	\left( 1- \frac{d}{m\log(2)} \right)^{\frac{m}{2}} \approx e^{-\frac d {2\log(2)}}\text.
\end{align*}
We conditioned on $v$ not matching at arrival, but joining the pool instead. This happens for approximately half of all agents. Thus, the total loss equals approximately $\frac{1}{2} e^{-\frac d {2\log(2)}}$.\\

\paragraph*{Patient algorithm} Again, we first want to control the pool size.  Assuming a negligible loss, for a fixed agent $v$, approximately $\frac{m}{2}$ agents get critical in the interval $\left[t_v, t_v+1\right]$, where $t_v$ is the arrival time of $v$. Under the assumption that the pool size stays constant over time, by symmetry, the probability that none of these agents matches with $v$ should equal
\begin{align*}
	\left(1-\frac{1}{z_{\mathit{Eq}}}\right)^{\frac{m}{2}} \approx e^{- \frac{m}{2\, z_{\mathit{Eq}}} } \overset{!}{=} \frac{1}{2}.
\end{align*}
Solving this yields $z_{\mathit{Eq}} \approx \frac{m}{2\log(2)}$. Now, the probability that an agent getting critical perishes equals approximately
\begin{align*}
	\left(1-\frac{d}{m}\right)^{z_{\mathit{Eq}}} \approx \left(1-\frac{d}{m}\right)^{\frac{m}{2 \log(2)}} \approx e^{-\frac{d}{2 \log(2)}}\text.
\end{align*}
Assuming the loss is small, approximately half of the agents get critical, i.e., they are not matched with an other agent before the end of their maximum \dtimename. Thus, the total loss among all agents equals approximately $\frac{1}{2} e^{-\frac d {2\log(2)}}$.
Hence, both arguments lead to an approximately equal loss.

\subsection{Heterogeneous Agents}

Many applications of dynamic matching markets comprise \emph{heterogenous} agents differing in their market behavior, for instance with respect to their urgency to be matched or with respect to their capability to form matches with other agent types. 
In this section, we give a brief overview of how to extend our results to markets with heterogeneous agents.

First, it is possible to extend our results to a setting with heterogeneous agents regarding their urgency to be matched, that is, agents differing with respect to their maximum \dtimenamepl.
A first extension is to consider a setting where agents have different guaranteed \dtimenamepl.
This allows us to directly extend our results: 
As long as all agents \emph{individually} satisfy the conditions on the maximum \dtimename, our results still apply. 
To be more explicit, consider, for example, the following scenario. 
There are two types of agents arriving to the market at random. 
The first type consists of \emph{urgent} agents, which only have a guaranteed maximum \dtimename of $\eps < 1$. 
The second type of agents consists of \emph{persistent} 
agents having a guaranteed maximum \dtimename of $1$.
We can model this with an urgency rate $q\in [0,1]$.
Agents still arrive with Poisson rate $m$ and any two agents are compatible with probability $p=\frac{d}{m}$, where $d$ is a \dparaname.
Additionally, upon their arrival, an agent is urgent with probability $q$ and persistent with probability $1-q$.
This probability is independent of other random events such as the arrival of an agent or the compatibility of an agent.
Hence, urgent agents arrive at the market at Poisson rate $m_1 = qm$, whereas persistent agents independently from urgent agents arrive at Poisson rate $m_2 = (1-q)m$.
As the \dtimename of every agent guarantees a maximum sojourn of $\epsilon$,
 \Cref{thm:epslowerbound} still holds and yields 
\begin{equation*}
    \loss[\gdy] \leq e^{-\frac{\eps d}{2\log(2)}}\text.
\end{equation*}

However, knowledge of the nature of the agents, i.e., their division into urgent and persistent agents allows us to derive a better loss bound.
In fact, we can extract a loss probability conditioned on a sojourn guarantee from the proof of \Cref{thm:greedy_loss_upper_bound}.

\begin{lemma}\label{lem:singleagentloss}
    Let $\eps >0$ and let $\ag_i$ be an agent with \dtimename $\sou_i$ arriving at time $t$.
    Then, under $\gdy$, it holds for large enough $m$ that 
    \begin{equation*}
        \p(\ag_i \text{ perishes}\mid \sou_i \ge \eps) \le e^{-\frac{\eps d}{2 \log(2)}}\text.
    \end{equation*}
\end{lemma}

\begin{proof}
    We define a good event as

    \begin{align*}
	\mathcal{G}_t := 
	&
	\left\{ \psize_t \leq  \frac{C_2 \log(2)}{d}m \right\}
	\cap \left\{ \psize_s \leq  \frac{C_3 \log(2)}{d}m \forall s\in \left[t,t+1\right] \right\}
	\cap \left\{ \numbarr_{\left[t,t+\eps\right]} \geq \eps(m- 30\log(m)m^{1/2})\right\} \text.
	\end{align*}

    This differs from the good event in the proof of \Cref{thm:greedy_loss_upper_bound} only with respect to the third part.
    Then, \eqref{eq:poolovertimebound} is adjusted as 
    \begin{align*}
		\p_\dep \left( \numbarr_{\left[t,t+\eps\right]} < \eps(m- 30\log(m)m^{1/2})\right) 
		&
		= \p_\dep \left( \numbarr_{\left[t,t+\eps\right]} < \eps m \left( 1- \frac{30\log(m)}{m^{1/2}} \right) \right)\\
		&
		\leq e^{- \frac{\eps m}{3} \left(\frac{30\log(m)}{m^{1/2}}\right)^2} \leq m^{-10}\text.
	\end{align*}

    Hence, the good event still applies with high probability.
    Moreover, \eqref{eq:singleaglossgood} turns into

	\begin{align*}
	\p(\ag_i \text{ perishes}\mid \mathcal{G}_t, \sou_i \ge \eps) \le \left( 1-\frac{d}{ 2 C_3 \log(2) m} \right)^{\eps(m-30\log(m)m^{1/2})} \underset{m\to \infty}{\longrightarrow} e^{-\frac{\eps d}{2 \log(2)}}\text.
	\end{align*}

    Inserting this into the analogue of \eqref{eq:singleagloss}, we obtain $\p(\ag_i \text{ perishes}\mid \sou_i \ge \eps) \le  e^{-\frac{\eps d}{2 \log(2)}}$ for $m$ large enough.
\end{proof}

We can get a refined loss bound for the heterogeneous market with urgent and persistent agents.

\begin{theorem}\label{thm:heterobounds}
    Consider the matching market with urgent and persistent agents that have a guaranteed \dtimename of $\eps$ and $1$, respectively.
    Then, for $d\geq 2$, we have
	\begin{equation*}
	   \loss[\gdy] \leq q \cdot e^{-\frac{\eps d}{2\log(2)}} + (1-q) \cdot e^{-\frac{d}{2\log(2)}}\text.
	\end{equation*}
\end{theorem}

\begin{proof}
    Consider an agent $\ag_i$ arriving at time $t$ with \dtimename $\sou_i$.
    By \Cref{lem:singleagentloss}, it holds that   
    \begin{align*}
        &\p\left(\text{an agent arriving at time $t$ perishes} \right) \\
        &= q \cdot \p\left(\text{an agent arriving at time $t$ perishes }\mid \sou_i \ge \eps\right) \\
        &\phantom{=} + (1-q)\cdot \p\left(\text{an agent arriving at time $t$ perishes }\mid \sou_i \ge 1\right) \\
        &\le q \cdot e^{-\frac{\eps d}{2\log(2)}} + (1-q) \cdot e^{-\frac{d}{2\log(2)}}\text.
    \end{align*}

    Inserting this to \Cref{lem:aux-loss} yields the desired bound.
\end{proof}

However, urgent agents might not even have any guaranteed \dtimename, i.e., they might have a very short sojourn with moderate probability.
As a second case, we consider the model where the urgent agents, instead of having a guaranteed sojourn time of $\eps$ have an exponentially distributed sojourn time with mean~$1$.
Hence, these agents might stay for a longer time, but they are possibly very urgent.
As before, we assume that an agent is urgent or persistent with probability $q$ and $1-q$, respectively.

In this case, high loss is inevitable.
However, it is only caused by the urgent agents.
In fact, for $0\le s < t\le \Tmax$, let $\arr^{P}_{\left[s,t\right]}$ denote the set of persistent agents arriving in the time interval $[s,t]$.
We define the \emph{loss of persistent agents} of algorithm $\alg$ as 
$$\loss^P := \limsup_{m,\Tmax\to \infty} \frac{\E[|\arr^P_{[0,\Tmax]} \setminus (\alg(m,\Tmax) \cup \pool_\Tmax)|]}{(1-q)m\Tmax}\text.$$

Then, while loss is inevitable, the loss of $\gdy$ among persistent agents is exponentially small. 

\begin{theorem}\label{thm:heteroexp}
    Consider the matching market with urgent and persistent agents, where the \dtimename of the former is exponentially distributed with mean~$1$ and the latter have a guaranteed \dtimename of~$1$.
    Then, $\gdy$ satisfies

    \begin{equation*}
        \loss[\gdy] \geq \frac{q}{12}\cdot \frac {1}{d} \quad\text{ and }\quad \loss[\gdy]^P\le e^{-\frac{d}{2 \log(2)}}\text.
    \end{equation*}
\end{theorem}

\begin{proof}
    As in the discussion after the proof of \Cref{thm:highloss}, we can see that the \dtimename where agents have an exponential \distname with mean~$1$ with probability $q$ satisfies the conditions of \Cref{thm:highloss} with $c = \frac q2$ and $\eps_0 = 1$.
    Hence, 
    \begin{equation*}
        \loss[\gdy] \geq \frac{q}{12}\cdot \frac {1}{d}\text.
    \end{equation*}

    In addition, following the proof of \Cref{lem:aux-loss}, one can derive
    \begin{align}
	\loss^P = \limsup_{m,\Tmax\to \infty} \frac 1{\Tmax} \int_{0}^{\Tmax} \p\left(\text{an agent $\ag$ arriving at time $t$ perishes}\mid \ag \text{ persistent}\right) dt \text.\label{eq:persistentloss}
	\end{align}

    By \Cref{lem:singleagentloss}, we have that 
    $$\p\left(\text{an agent $\ag$ arriving at time $t$ perishes}\mid \ag \text{ persistent}\right) \le e^{-\frac{d}{2 \log(2)}}\text.$$
    
    Inserting this into \eqref{eq:persistentloss}, we obtain the desired loss bound for persistent agents.
\end{proof}

Finally, we want to briefly discuss another type of heterogeneity considered in the literature that differentiates easy-to-match and hard-to-match agents \citep{ABJM19a,ANS22a}, which are distinguished in terms of compatibility. 
In our setting, this can be modeled by having agents arrive with Poisson rate $m$ where membership to a class is determined by a random variable. 
Then, the compatibility of agents is decided according to different \dparanamepl $d_{\mathit{HH}}$, $d_{\mathit{EH}}$, and $d_{\mathit{EE}}$ with $d_{\mathit{HH}} \le d_{\mathit{EH}} \le d_{\mathit{EE}}$ representing the density of edges between two hard-to-match, two different-type, and two easy-to-match agents, respectively. 
This leads to matching probabilities $p_{\mathit{HH}} = d_{\mathit{HH}}/m$, $p_{\mathit{EH}} = d_{\mathit{EH}} / m$, and $p_{\mathit{EE}} = d_{\mathit{EE}} / m$. 
It is possible to extend our results under certain assumptions on the relationship between the different \dparanamepl. Still, a rigorous treatment of this model is beyond the scope of this paper, and we leave it as an intriguing direction for further research. 

\Cref{thm:heterobounds,thm:heteroexp} give insights about the loss of the $\gdy$ algorithm in settings with heterogeneous agent types.
Using $\gdy$ in these settings is useful because it maintains its other advantages.
For example, our results on the waiting time still apply.
By \Cref{lem:agg_waiting}, we obtain an inverse linear average waiting time, which is close to optimal for most $d$ according to \Cref{prop:waitlb}.
However, a reasonable approach when dealing with heterogeneous agents is to algorithmically react to their type.
For example, following the strong performance of $\gdy$ for a guaranteed \dtimename and of $\pat$ for exponentially distributed \dtimenamepl, one could match the former greedily and the latter patiently.
This gives rise to study hybrid algorithms that meet agent demands by use different algorithms for different agent types.
We leave their consideration as an interesting avenue for future work.

\section*{Acknowledgements}
An extended abstract of this article appeared in the 
Proceedings of the 24th ACM Conference on Economics and Computation (2023).
This work was supported by the Deutsche Forschungsgemeinschaft (German Research Foundation) under grants {BR~2312/11-2}, {BR~2312/12-1}, and {277991500/GRK2201}. 
We thank Matthias Greger, Shunya Noda, and Peng Shi as well as the reviewers from ACM EC and JET for valuable feedback. 
The source code for the simulations is publicly available at \url{https://github.com/stefan-kober/dynamicmatching}.

\clearpage
\appendix

\section*{Appendix}

In the appendix, we present proofs of technical and auxiliary lemmas.

\section{Technical Lemmas in \Cref{sec:lossguarantee}}\label{app:greedyloss}

This section contains the proofs of technical lemmas in our analysis of the loss guarantee of the greedy algorithm. 
First, we prove our coupling lemma.

\coupling*

\begin{proof}
	We consider the procedure where we have two pools to which two identical copies of agents arrive according to the same Poisson process of rate $m$. 
    The first pool acts according to the greedy algorithm with \dtimename $\sou$ and initial pool $\pool_0$ and the second one according to the greedy algorithm under $\sou_{\infty}$ and initial pool $\pool_0^\infty$.
    Let  $\dep_\infty$ be the distribution of $\sou_{\infty}$, which is given by $\dep_\infty(\{+\infty\}) = 1$.
	Recall that, attached to each arriving agent~$\ag$, there is a random vector $U(\ag) \in \{0,1\}^{\N}$ storing the random compatibility with other agents, as described in \Cref{sec:model}. In particular, if the pool is of size $K$ and an agent $\ag$ arrives, then the agent directly forms an edge if and only if $U_j(\ag) = 1$ for some $j \in \{1,\ldots, K\}$. If there are compatible agents in the pool, then the agent $\ag$ selects a partner among them chosen uniformly at random.
    Hence, we couple the two systems by using the same arrival times (one Poisson process of rate $m$), the same compatibility vectors $U(\ag)$ for each arriving agent $\ag$, and the same tie-breaking vectors. The only difference is the \dtimename: $\sou$ versus $\sou_\infty$.

	We show the assertion \eqref{eq:coupling1} by a case distinction. The statement is clear for the case $t=0$. For $t > 0$, the only interesting times are the ones where agents arrive, as agents can only perish from the first pool. Hence, $\psize_t^\infty$ cannot get smaller without the arrival of a new agent.

	If $\psize_t<\psize_t^\infty$ and an agent arrives at time $t$, then $\psize_t$ can increase by at most $1$ and $\psize_t^\infty$ can decrease by at most $1$, so \eqref{eq:coupling1} is still satisfied after the arrival of the agent. 
	If $\psize_t = \psize_t^\infty$ and an agent arrives at time~$t$, then the agent matches under the measure $\dep$ if and only if she matches under the measure~$\dep_\infty$.
    This holds because in both systems the decision to match for an arrival at $t$ depends only on whether one of the first $K$ components of $U(\ag)$ equals~$1$; these components are shared across the two systems by our coupling.
    So \eqref{eq:coupling1} is also satisfied after the arrival of the agent. If $\psize_t = \psize_t^\infty +1$ and an agent arrives, then if the agent does not match with an element in the first pool, she also does not match with an element in the second pool, so both pools increase in size.
	So it is not possible that $\psize_t$ gets larger while $\psize_t^\infty$ gets smaller at the same time which shows that \eqref{eq:coupling1} is always satisfied.
\end{proof}

Now, we prove the two lemmas regarding the pool size.

\greedyUbPool*

\begin{proof}
	Let $j\geq \frac{C_1 \log(2) m}{d}$. Then, we have
	\begin{align*}
	p(j,j+1)=\left(1-\frac{d}{m}\right)^j \leq \left(1-\frac{d}{m}\right)^\frac{C_1 \log(2) m}{d} \leq \frac{1}{2} e^{-\frac{10}{\log(m)}}\text.
	\end{align*}
	On the other hand, we also have that $p(j,j-1)\geq 1/2$. Inserting this into \eqref{eq:balance_equations}, we get
	\begin{align*}
	\rho(j+1) = \rho(j) \frac{p(j,j+1)}{p(j+1,j)} \leq e^{-\frac{10}{\log(m)}} \rho(j)\text.
	\end{align*}
	Thus, we get inductively
	\begin{align*}
	\rho\left( \left\lceil \frac{C_1 \log(2) m}{d} \right\rceil + n \right) \leq 
	e^{-n\frac{10}{\log(m)}} \rho\left( \left\lceil \frac{C_1 \log(2) m}{d} \right\rceil\right)  \leq e^{-n\frac{10}{\log(m)}}\text.
	\end{align*}
	Summing this over different values of $n$, we get for $m$ large enough
	\begin{align*}
	\sum_{k > \frac{C_1 \log(2) m}{d} + \frac 32 \log(m)^2} \rho(k) & \leq \sum_{k=\lceil \log(m)^2 \rceil }^{\infty} \rho\left( \left\lceil \frac{C_1 \log(2) m}{d} \right\rceil + k \right) 
	\leq \sum_{k=\lceil\log(m)^2\rceil}^{\infty} e^{-k\frac{10}{\log(m)}}\\
	&
	= e^{-\lceil \log(m)^2 \rceil \frac{10}{\log(m)}} \sum_{k=0}^{\infty} e^{-k\frac{10}{\log(m)}}\\
	&
	\leq m^{-10} \frac{1}{1-e^{-\frac{10}{\log(m)}}} 
	\leq m^{-10} \frac{\log(m)}{5} \leq m^{-9}\text.
	\end{align*}
	There, we used that $1-e^{-\frac{10}{\log(m)}} \geq \frac{5}{\log(m)}$ for $m$ large enough in the second to last step.
\end{proof}

\poolgdyub*

\begin{proof}
	We seek an upper bound for the probability that the pool size reaches $C_3 \frac{\log(2)}{d} m$ when starting at a size of at most $C_2 \frac{\log(2)}{d} m$. Let $M=\lceil C_2 \frac{\log(2)}{d} m \rceil$. When the pool has a size of at least $M$ and an agent arrives, then the probability that the agent joins the pool, i.e., that she is not matched directly, equals
	\begin{equation}\label{eq:drift}
	\left( 1-\frac{d}{m} \right)^{M} \leq \left( 1-\frac{d}{m}\right)^{C_2 \frac{\log(2)}{d} m} 
	\leq e^{- C_2 \log(2)}
	\leq \frac{1}{2} e^{-\frac{10}{\log(m)}}.
	\end{equation}
    
	There, we used \Cref{lem:exp-estimate} in the second inequality.
	Whenever a new agent arrives, the pool size can increase or decrease by $1$, where we have a tendency towards a decrease of the pool size. Thus, in particular, if the pool starts at size $M$, the probability that the pool size increases beyond $C_3 \frac{\log(2)}{d} m$ before decreasing to $C_2 \frac{\log(2)}{d} m$ is bounded by the probability of this event for the case of a random walk with drift to the left of at least $\frac{1}{2} e^{-\frac{10}{\log(m)}}$. The Euclidean distance between $C_2 \frac{\log(2)}{d} m$ and $C_3 \frac{\log(2)}{d} m$ is $2\log(m)^2$. Hence, there are more than $\log(m)^2$ integer points between these values, for $m\geq 3$.
	We can apply \Cref{lem:random walk with drift} with $\frac{1}{2}-\eps = \frac{1}{2}e^{-\frac{10}{\log(m)}}$ in order to bound the probability that the Markov chain with drift crosses such an interval. Thus, the probability of this event is bounded by 
	\begin{equation*}
	\left( 2 \frac{1}{2} e^{-\frac{10}{\log(m)}} \right)^{\log(m)^2} = m^{-10}.
	\end{equation*}
	If the pool size starts at some value less than $\frac{C_2 \log(2)}{d} m$ and then grows to some value  greater than $\frac{C_3 \log(2)}{d} m$ then at some time interval in between the pool size starts with size $M$ and reaches size $\frac{C_3 \log(2)}{d} m$ before going back to $M-1$. 
	
	Assume that $k$ agents arrive during the time interval $\left[ 0,1 \right]$, i.e., $\numbarr_{\left[0,1\right]}=k$. For $i\in \{1,\ldots,k\}$, the probability that the pool size is $M$ when the $i$th agent arrives and then increases beyond $\frac{C_3 \log(2)}{d} m$ before going to $M-1$ is bounded by $m^{-10}$ by the previous calculations. Thus, by a union bound, the probability that the pool size increases beyond $\frac{C_3 \log(2)}{d} m$ is bounded by $km^{-10}$. Conditioning on the number of arrivals in the time interval $\left[ 0,1 \right]$, we get
	\begin{align*}
	&\p_{\dep}\left( \psize_t > \frac{C_3 \log(2) m }{d} \text{ for some } t \leq 1 \  \big| \ \pool_0=\pool  \right) \\
	&
	= \sum_{k=1}^{\infty }
	\p_{\dep}\left( \psize_t > \frac{C_3 \log(2) m }{d} \text{ for some } t \leq 1 \  \big| \ \pool_0=\pool, \numbarr_{\left[0,1\right]=k}  \right) \p \left( \numbarr_{\left[0,1\right]}=k \right)\\
	&
	\leq
	\sum_{k=1}^{\infty } k m^{-10} \p \left( \numbarr_{\left[0,1\right]}=k \right) = m^{-9}
	\end{align*}
	as $\numbarr_{\left[0,1\right]}$ is a Poisson random variable with expectation value $m$.
\end{proof}

Next, we provide the technical lemma that allows us to compute the loss by integrating over the probability of perishing.

\auxloss*

\begin{proof}
	In the proof, we omit the subscript $\dep$, because the distribution of the \dtimenamepl is fixed.
	Recall that $\sou_i$ denotes the maximum \dtimename of agent $\ag_i$. In \Cref{sec:model}, we described that every agent $\ag$ carries a random variable $U(\ag)$ that determines with which agents in the pool she forms an edge at arrival. 
    We also assume that attached to agent $\ag$ there is a random variable $\tiebreak(\ag)$, which has the distribution $\text{Uniform}\left[0,1\right]^{\otimes \N}$. 
    Whenever an agent $\ag$ decides to match with another agent, but has several possibilities, it chooses the agent $\ag_i$, where $\tiebreak_i(\ag)$ is minimal. As the uniform distribution is absolutely continuous with respect to the Lebesgue measure, there is almost surely a unique minimizer. 
    By symmetry, the agent that $\ag$ matches with is chosen uniformly among all possibilities. So $\tiebreak(\ag)$ helps with breaking ties, in case there are some. 
	For the complete process we have four different sources of randomness: The incoming agents arrive at random, according to a Poisson process, their \dtimenamepl are random with distribution $\mu$, the forming of edges which is determined by the vectors $U(\ag)$, and the random breaking of ties, if necessary, that is determined by the vectors $\tiebreak(\ag)$.
	Let $Z = \left((X_i, U(\ag_i), \tiebreak(\ag_i))\right)_{i\ge 1}$ be the random variable containing all information about the random variables $X_i, U(\ag_i)$ and $\tiebreak(\ag_i)$ for $i\ge 1$. Given $Z$, the evolution of an agent $v_i$ up to time $\Tmax$ depends only on the random variable $Y = (Y_t)_{t\in[0,\Tmax]}$ distributed according to the increments of a Poisson process with rate $m$, where $Y_t$ is the number of agents arriving at time $t$. 
    Thus, almost surely, we have $Y_{t} \in \{0,1\}$ for all $t$, with $Y_t=1$ if and only if an agent arrives at time $t$.
	Using $Y$ and $Z$, we can define the function 
	\begin{align*}
		f_Z(t,Y) \coloneqq \mathbf{1} \left\{ \text{an agent arrives at time $t$ and perishes before time $\Tmax$}  \right\}\text, 
	\end{align*}
	where $\mathbf 1$ is the indicator function.
	The loss until time $\Tmax$ can clearly be expressed as $\loss(m,\Tmax) = \frac 1{m\Tmax}\E  \left[ \int_0^{\Tmax} f_Z(t,Y) Y(dt) \right]$.
	Note that $Z$ and $Y$ are independent, so the process $Y$ is still the same Poisson process, given the information contained in $Z$. Furthermore, the intensity measure of this Poisson process is simply $m$ times the Lebesgue measure on $\left[0,\Tmax\right]$. 
    Thus, we can apply the tower property in the first and fourth equality and Mecke's equation \citep[see, e.g.,][Theorem~4.1]{LaPe17a} in the second equality to compute
	\begin{align*}
		  \E & \left[ \int_0^{\Tmax} f_Z(t,Y) Y(dt)  \right]\\
		 & = \E \left[\E \left[ \int_0^{\Tmax} f_Z(t,Y) Y(dt) \Big| \sigma(Z)  \right] \right]
		 = \E \left[ \int_0^{\Tmax} \E \left[  f_Z(t,Y+ \delta_t) \Big| \sigma(Z) \right] m dt \right]\\
		 & = m \int_0^{\Tmax} \E \left[  \E \left[  f_Z(t,Y+ \delta_t)  \Big| \sigma(Z)  \right]\right] dt
		 = m \int_0^{\Tmax} \E \left[    f_Z(t,Y+ \delta_t)   \right] dt \\
		 & = m \int_0^{\Tmax} \p\left( \text{an agent arriving at time $t$ perishes before time $\Tmax$}  \right)  dt\text.
	\end{align*}
	There, $\sigma(Z)$ denotes the $\sigma$-algebra generated by $Z$, and $\delta_t$ is the Dirac delta function, i.e., $\delta_t(s)=\mathbf{1}_{\{s=t\}}$. Now, let us consider the second statement of the lemma. By the equality that we proved just before, we know that
	\begin{align}\label{eq:substituion}
		\notag \loss & = \limsup_{m,\Tmax\to \infty} \frac 1{\Tmax} \int_{0}^{\Tmax} \p_{\dep}\left(\text{an agent arriving at time $t$ perishes before time $\Tmax$} \right) dt \\
		& = \limsup_{m,\Tmax\to \infty} \int_{0}^1 \p_{\dep}\left(\text{an agent arriving at time $s\Tmax$ perishes before time $\Tmax$} \right) dt \text.
	\end{align}
	where the last equality holds by integration by substitution. For $\Tmax>0$, we define the functions $\tilde{f}_{\Tmax}, f_{\Tmax} : \left[0,1\right] \rightarrow \left[0,1\right]$ by
	\begin{align*}
		f_{\Tmax}(s) & \coloneqq \p_{\dep}\left(\text{an agent arriving at time $s\Tmax$ perishes before time $\Tmax$} \right)\text, \\
		\tilde{f}_{\Tmax}(s) & \coloneqq \p_{\dep}\left(\text{an agent arriving at time $s\Tmax$ perishes} \right) \text.
	\end{align*}
	We clearly have $f_t(s) \leq \tilde{f}_{\Tmax}(s) $ for all $s \in \left[0,1\right]$, and for fixed $s<1$ and $m>0$ we also have
	\begin{align*}
		0 \leq \tilde{f}_{\Tmax}(s) - f_{\Tmax}(s) &=
		\p_{\dep} \left(\text{an agent arriving at time $s\Tmax$ perishes after time $\Tmax$}\right)
		\\
		& \leq \p_{\dep} \left(\text{the maximum sojourn of an agent is in $\left[(1-s)\Tmax,\infty\right)$}\right)\\
		& = \mu \left(\left[(1-s)\Tmax,\infty\right)\right)
	\end{align*}
	which converges to $0$ as $\Tmax$ tends to infinity, since $\mu(\{+\infty\})=0$. As the difference $f_{\Tmax}(s) - \tilde{f}_{\Tmax}(s)$ is bounded by $1$ for all $s \in \left[0,1\right]$ we get that
	\begin{align*}
		\lim_{\Tmax \to \infty} \int_0^1 f_{\Tmax}(s) - \tilde{f}_{\Tmax}(s) ds = 0
	\end{align*}
	and as this convergence holds uniformly for all $m$, we get 
	\begin{align*}
		\loss & = \limsup_{m,\Tmax\to \infty} \frac 1{\Tmax} \int_{0}^{\Tmax} \p_{\dep}\left(\text{an agent arriving at time $t$ perishes before time $\Tmax$} \right) dt\\
		& = \limsup_{m,\Tmax\to \infty} \frac 1{\Tmax} \int_{0}^{\Tmax} \p_{\dep}\left(\text{an agent arriving at time $t$ perishes} \right) dt \text.
	\end{align*}
    
    This completes the proof.
\end{proof}

\section{Analysis of the Patient Algorithm}\label{app:patient}

In this appendix, we provide the proof of the technical lemmas in the proof of \Cref{lem:patient_lb}.
Note that the statements of the lemmas refer to the notation of this proof.

\stochdom*

\begin{proof}
	We first describe three different experiments: In the first one, we have a pool $Z_{t-\frac{4}{3}}$ of some size. We call the agents in the pool at time $t-\frac{4}{3}$ the uncolored {\sl black agents}. These agents get critical exactly one time step after their arrival, so they will not be in the pool at time $t-\frac{1}{3}$ anymore. The $k_1$ many agents arriving in the interval $\left[t-\frac{4}{3} , t-1 \right)$ are called the {\sl red agents} and the $n=k_2+k_3$ many agents arriving in the interval  $\left[ t-1, t-\frac{1}{3} \right)$ are called the {\sl blue agents}. Then agents can depart from the pool by perishing or by matching, so that $l$ colored, i.e., red or blue, agents are in the pool at time $t-\frac{1}{3}$. We are interested in the number $K_1$, which is the number of red agents in the pool at time $\left[t-\frac{1}{3}\right]$.
    
	In the second experiment, we have an urn with $k_1$ red balls and $n$ blue balls and we take $k_1+n-l$ balls out of the urn, uniformly at random without replacement. We are interested in the number of red balls $K_1^\prime$ that are left in the urn after this procedure.
    
	In the third experiment, we have an urn with $k_1$ red balls and $n$ blue balls and we take $l$ balls out of the urn, uniformly at random without replacement. We are interested in the number of red balls $\tilde{K_1}$ that are taken out of the urn.
	
	Uniformly choosing $l$ balls to be left in an urn or uniformly choosing $l$ balls to be taken out of an urn gives rise to the exact same distribution. 
    Thus it suffices to show that $K_1^\prime$ stochastically dominates $K_1$. For this, we couple the urn to the pool in the following way: At time $t-\frac{4}{3}$, we start with the initial urn, i.e., the one with $k_1$ red balls and $n$ blue balls in it. Let $\left(R_s\right)_{s\in \left[t-\frac{4}{3}, t-\frac{1}{3}\right]}$, respectively $\left(B_s\right)_{s\in \left[t-\frac{4}{3}, t-\frac{1}{3}\right]}$, be the number of red, respectively blue, agents in the pool at time $s$. Analogously, let $\left(R_s^\prime\right)_{s\in \left[t-\frac{4}{3}, t-\frac{1}{3}\right]}$, respectively $\left(B_s^\prime\right)_{s\in \left[t-\frac{4}{3}, t-\frac{1}{3}\right]}$, be the number of red, respectively blue, balls in the urn at time $s$. The underlying algorithm of the matching process in the first experiment does not distinguish between red and blue agents. So whenever an agent gets critical and matches with a colored agent, only the proportion of red/blue agents matters. So we can think of the process of the pool as follows: When an agent gets critical, say at time $t_0$, and matches with a colored agent, there is an underlying random variable $U_{t_0}\sim \text{Uniform}\left[0,1\right]$ that determines whether the agent matches with a blue or with a red agent. When $U_{t_0}\leq \frac{R_{t_0}}{R_{t_0}+B_{t_0}}$, then we choose uniformly at random one of the $R_{t_0}$ many red agents. Otherwise we choose one of the $B_{t_0}$ many blue agents uniformly at random.
	Whenever there is a matching with a colored agent in the pool, we also remove a ball from the urn. We pick a red ball if $U_{t_0}\leq \frac{R_{t_0}^\prime}{R_{t_0}^\prime + B_{t_0}^\prime }$, and otherwise we pick a blue ball. As $U_{t_0}$ is the uniform distribution, this also gives rise to the correct probabilities. Note that the proportion of red agents among the colored agents in the pool can change either when agents depart by matching or when a new agent arrives. Contrary, the proportion of red balls in the urn changes only at times when colored agents match in the pool. Let $S_s$, respectively $S_s^\prime$ be the number of red agents, respectively balls, that are removed from the pool, respectively urn, by time $s$. We claim that
	\begin{align}\label{eq:S}
		S_s \geq S_s^\prime \text{ for all } s\in \left[t-\frac{4}{3}, t-\frac{1}{3} \right] \ .
	\end{align}
	This is clear for $s < t-1$, as $S_s$ is simply the number of removed colored agents, as no blue agent entered the pool yet. Thus we can assume $s\geq t-1$ from here on. Whenever $S_s>S_s^\prime$ just before some time $t_0$, at which an agent gets critical and matches with a colored agent, then $S_s\geq S_s^\prime$ still holds after this match. If $S_s = S_s^\prime$ and $s\geq t-1$, then all red agents already arrived, which gives $R_s = k_1- S_s = R_s^\prime$. As the number of colored balls in the urn is at least as high as the number of colored agents in the pool, we have
	\begin{align*}
		\frac{R_s^\prime}{R_s^\prime + B_s^\prime} \leq  \frac{R_s}{R_s + B_s} \text.
	\end{align*}
	Assume that $S_s=S_s^\prime$ just before a match with a colored agent in the pool occurs. If a red ball is drawn from the urn, i.e., $U_s \leq \frac{R_s^\prime}{R_s^\prime + B_s^\prime}$, then also a red agent leaves the pool in a matching. Thus $S_s \geq S_s^\prime$ still holds directly after the match. So we also get that $S_{t-\frac{1}{3}} \geq S_{t-\frac{1}{3}}^\prime$ which directly gives
	\begin{align*}
		K_1^\prime = R^\prime_{t-\frac{1}{3}} = k_1 - S^\prime_{t-\frac{1}{3}} 
		\geq
		k_1 - S_{t-\frac{1}{3}} = R_{t-\frac{1}{3}}
	\end{align*}
	and thus finishes the proof.
\end{proof}

\urnbound*

\begin{proof}
	Consider a set $U \subseteq \{1,\ldots,l\}$.
	Then, the probability of drawing a red ball at the $i$th draw if and only if $i \in U$ is given by
	\begin{align*}
		\prod_{i=0}^{|U|-1} \frac{k_1-i}{N-i} \prod_{i=0}^{l-|U|-1} \frac{n-i}{N-|U|-i} = \frac{1}{\prod_{i=0}^{l-1} N-i } \prod_{i=0}^{|U|-1} (k_1 - i) \prod_{i=0}^{l-|U|-1} (n-i) \text.
	\end{align*}
	Among all $U \subseteq\{1,\ldots,l\}$ with $|U|\geq \frac{l}{2}$, this is clearly maximized when $|U|=\lceil \frac{l}{2} \rceil$, as $n\geq k_1$. Let us assume that $l$ is even from here on. The proof for $l$ odd works completely similar, just with one additional term. So we have for all $l$ even and all $U \subseteq\{1,\ldots,l\}$ with $|U|\geq \frac{l}{2}$ the bound
	\begin{align}\label{eq:product}
			\notag & \frac{1}{\prod_{i=0}^{l-1} N-i } \prod_{i=0}^{|U|-1} (k_1 - i) \prod_{i=0}^{l-|U|-1} (n-i)
			\leq \frac{1}{\prod_{i=0}^{l-1} N-i } \prod_{i=0}^{l/2-1} (k_1 - i) \prod_{i=0}^{l/2-1} (n-i)
			\\
			& =  \prod_{i=0}^{l/2 - 1} \frac{(k_1-i)(n-i)}{(N-2i)(N-2i-1)} 
			= 
			\prod_{i=0}^{l/2 - 1} \frac{(k_1-i)(n-i)}{(N-2i)(N-2i)} \prod_{i=0}^{l/2-1} \frac{N-2i}{N-2i-1}\text.
	\end{align}
	The second factor in \eqref{eq:product} can be bounded by
	\begin{align*}
		\prod_{i=0}^{l/2-1} \frac{N-2i}{N-2i-1} \leq \prod_{i=0}^{l/2-1} \frac{N-2i+1}{N-2i-1} = \frac{N+1}{N-l+1} \leq N+1 \ .
	\end{align*}
	In order to bound the components of the first factor, first notice that
	\begin{align*}
		\frac{k_1-i}{N-2i}+\frac{n-i}{N-2i} = \frac{ N-2i}{N-2i} = 1\text,
	\end{align*}
	so we multiply two numbers that sum to one. Furthermore, $\frac{k_1-i}{N-2i} \leq  \frac{2N/5-i}{N-2i} = \frac 25 \frac{N-5i/2}{N-2i} \leq \frac 25$. Recall that we assumed that $\frac {k_1}N\le \frac 25$, which is used in the first inequality. The product of $z(1-z)$ with $z \in \left[0,\frac 25\right]$ is maximized for $z=\frac 25$. Thus we already have
	\begin{align*}
		\prod_{i=0}^{l/2 - i} \frac{(k_1-i)(n-i)}{(N-2i)(N-2i)} \leq \left(\frac 25\cdot\frac 35\right)^{l/2} \leq 0.24^{l/2}.
	\end{align*}
	Inserting this into \eqref{eq:product} we get that for all $U\subseteq\{1,\ldots,l\}$ with $|U|\geq l/2$
	\begin{align*}
		\p \left( \text{The $i$th draw is red if and only if $i\in U$} \right) \leq (N+1)0.24^{l/2}.
	\end{align*}
	Recall that we have $l\geq m/8$ and, due $k_i\le m/3 + m^{2/3}$ for $i \in \{1,2,3\}$, $N+1 = k_1+k_2+k_3 +1 \leq 2m$ for $m$ large enough. By a union bound, we thus get
	\begin{align*}
		\p\left( \tilde{K_1} \geq \frac{l}{2} \right) & = \sum_{\substack{U \subseteq \{1,\ldots,l\}: \\ |U|\geq l/2}}
		\p \left( \text{The $i$th draw is red if and only if $i\in U$} \right) \leq 2^l(N+1)0.24^{l/2} \\
		& \leq 2m  \left(2\sqrt{0.24}\right)^l \leq 2m \, 0.98^l \leq 2m \, 0.98^{m/8}\text.
	\end{align*}
    
This concludes the proof of the lemma.
\end{proof}

\section{Auxiliary Statements}

In this section we give several auxiliary statements that we used above.

The first lemma is an auxiliary inequality that follows from the inequality $1+x\leq e^x$ with $x = - \frac{c}{m}$.

\begin{lemma}\label{lem:exp-estimate}
	Let $c\geq 0$ and $m>c$. Then,
	\begin{equation}\label{eq:Inequality for exponential approximation}
	\left( 1- \frac{c}{m} \right)^m \leq e^{-c}\text.
	\end{equation}
\end{lemma}

\begin{lemma}\label{lem:exp-snd-estimate}
    For all $x\ge 0$, it holds that $1- e^{-x} \ge \frac x{1+x}$.
\end{lemma}

\begin{proof}
    Taking reciprocals in the inequality $1+x\leq e^x$ (note that both sides are positive for $x\ge 0$) yields $e^{-x} \le \frac 1{1+x}$.
    Hence,
    \begin{equation*}
        1- e^{-x} \ge 1- \frac 1{1+x} = \frac x{1+x}\text.
    \end{equation*}
\end{proof}

An application of Chernoff's inequality \citep[see, e.g.,][]{HaRu90a} yields the next result.

\begin{lemma}\label{lem:Chernoff}
	Let $X$ be the sum of independent Bernoulli random variables or a Poisson random variable with expectation value $\mu$ and let $0<\delta \leq 1$. Then,
	\begin{align}
	\label{eq:chernoff1}
	& \p\left( X \geq (1+\delta)\mu \right) \leq e^{-\mu\delta^2/3}, \\
	\label{eq:chernoff2}
	& \p\left( X \leq (1-\delta)\mu \right) \leq e^{-\mu\delta^2/3}\text.
	\end{align}
\end{lemma}

The next lemma is a typical application of Markov chains to analyze random walks with drift, and its proof uses a classical technique.

\begin{lemma}\label{lem:random walk with drift}
	Let $0<\eps<\frac{1}{2}$ and let $(X_n)_{n\in \N}$ be an irreducible Markov chain with state space $\N$, that makes jumps between nearest neighbours only and with transition probabilities satisfying 
    $p(z,z+1)\leq\frac{1}{2}-\eps$ for all $z\geq M$. 
    When the Markov chain starts at $X_0=M$, the probability of hitting some integer $N>M$ before going to $M - 1$ is bounded by $(1-2\eps)^{N-M}$.
\end{lemma}

\begin{proof}
	We provide a proof for the case $M = 1$. 
	The proof for general $M$ works completely analogous.
	
	We define the function 
	\begin{equation*}
	f:\N \rightarrow \R \ , \ n \mapsto 
	\sum_{i=0}^{n-1} \prod_{j=1}^{i} \frac{p(j,j-1)}{p(j,j+1)}\text,
	\end{equation*}
	i.e.,
	\begin{equation*}
		f(0)=0, f(1)=1, \text{ and } f(n+1)= f(n) + \prod_{j=1}^{n} \frac{p(j,j-1)}{p(j,j+1)} \text{ for } n\geq 1\text.
	\end{equation*}
    
	The function $f$ is harmonic on $\N_{>0}$ for the transition probabilities $\left(p(k,l)\right)_{k,l \in \N}$, as 
	\begin{align*}
		& p(n,n+1)f(n+1) + p(n,n-1)f(n-1) \\
		&
		= 
		p(n,n+1) \left(f(n)+\prod_{j=1}^{n} \frac{p(j,j-1)}{p(j,j+1)}\right)
		+
		p(n,n-1) \left(f(n)-\prod_{j=1}^{n-1} \frac{p(j,j-1)}{p(j,j+1)}\right)\\
		& = f(n) + p(n,n+1) \prod_{j=1}^{n} \frac{p(j,j-1)}{p(j,j+1)}- p(n,n-1)\prod_{j=1}^{n-1} \frac{p(j,j-1)}{p(j,j+1)}
		= f(n)
	\end{align*}
	for all $n\geq 1$. For an integer $t\in \Z$, define the stopping time
	$\tau_t=\min\{n:X_n=t\}$, and define $\tau=\tau_0\wedge \tau_N=\min\{n:X_n\in \{0,N\}\}$.
	Assume that the Markov chain starts at $X_0=1$. Then
	the process $f(X_{n\wedge \tau})$ is a martingale with respect to the natural filtration. As the state space is finite, this martingale is also uniformly integrable. By the optional sampling theorem \citep[][Theorem~10.21]{Klen13a}, we have
	\begin{align}\label{eq:martingale1}
	\notag 1 & = \E \left[ f(X_\tau) | X_0 =1 \right] 
	= 
	f(0) \cdot \p\left( X_\tau = 0 | X_0 =1 \right) + f(N) \cdot \p\left( X_\tau = N | X_0 =1 \right)\\
	&
	=
	f(N) \cdot \p\left( \tau_N < \tau_0 | X_0 =1 \right) \text.
	\end{align}
    
	Next, we bound $f(N)$ by 
	\begin{align*}
		f(N) & = \sum_{i=0}^{N-1} \prod_{j=1}^{i} \frac{p(j,j-1)}{p(j,j+1)} \geq \sum_{i=0}^{N-1} \prod_{j=1}^{i} \frac{\frac{1}{2}+\eps}{\frac{1}{2}-\eps}
		=
		\sum_{i=0}^{N-1} \left( \frac{1+2\eps}{1-2\eps} \right)^i
		= \frac{ \left( \frac{1+2\eps}{1-2\eps} \right)^N - 1}{ \frac{1+2\eps}{1-2\eps} - 1 }\\
		& =	\frac{ \left( \frac{1+2\eps}{1-2\eps} \right)^N - 1}{ \frac{4\eps}{1-2\eps} }
		= \frac{ \left( 1+2\eps \right)^N - (1-2\eps)^N}{4\eps (1-2\eps)^{N-1}} \geq  \frac{ 4\eps }{4\eps (1-2\eps)^{N-1}} =  \frac{ 1 }{ (1-2\eps)^{N-1}} \text.
	\end{align*}
    
	Inserting this inequality into \eqref{eq:martingale1} and solving for $\p\left( \tau_N < \tau_0 | X_0 =1 \right)$ gives that
	\begin{align*}
	\p\left( \tau_N < \tau_0 | X_0 =1 \right) \leq (1-2\eps)^{N-1} \text.
	\end{align*}

    This concludes the proof.
\end{proof}

\end{document}